%% file: path_sums.tex
\documentclass[]{llncs}
\usepackage{amssymb}
\usepackage{amsfonts}
\usepackage{mathtools}
\usepackage{calc}
\usepackage{multicol}
\usepackage{hyperref}
\hypersetup{
    colorlinks = true,
    citecolor = blue
}
\usepackage{graphicx}
\usepackage{needspace}

\usepackage{amsmath}

\usepackage{stmaryrd}
\usepackage{xspace}

\usepackage{geometry}
\usepackage{tikz}
\usepackage{pgfplots}
\usetikzlibrary{decorations.markings}
\usetikzlibrary{matrix,backgrounds,folding}
\usetikzlibrary{shapes.geometric, shapes.misc}
\pgfdeclarelayer{edgelayer}
\pgfdeclarelayer{nodelayer}
\pgfsetlayers{background,edgelayer,nodelayer,main}
\tikzstyle{every picture}=[baseline=-0.25em]
\tikzstyle{none}=[inner sep=0mm]
\tikzstyle{zxnode}=[shape=circle, minimum width=.25cm, inner sep=0.5pt, font=\footnotesize, draw=black]
\tikzstyle{gn}=[zxnode ,fill=green]
\tikzstyle{rn}=[zxnode ,fill=red]
\tikzstyle{H box}=[rectangle,fill=yellow,draw=black,xscale=1,yscale=1,font=\footnotesize,inner sep=1.2pt,minimum width=0.15cm,minimum height=0.15cm]
\tikzstyle{ug}=[regular polygon, regular polygon sides=3, fill=red,draw=black,inner sep = 0pt,minimum width=1em]
\tikzstyle{black dot}=[inner sep=0.7mm,minimum width=0pt,minimum height=0pt,fill=black,draw=black,shape=circle]
\tikzstyle{dot}=[black dot]
\tikzstyle{white dot}=[dot,fill=white]
\tikzstyle{box}=[rectangle,fill=white,draw=black, font=\scriptsize, inner sep=2pt]
\tikzstyle{zwcross}=[diamond, draw, fill=gray, minimum width=0em, inner sep=1.5pt]
\tikzstyle{arrow}=[decoration={markings,mark=at position 1 with
    {\arrow[scale=1.2,>=stealth]{>}}},postaction={decorate}]

\tikzstyle{st}=[star,star points = 5, fill=white,draw=black,inner sep = 1.2pt,line width=1.2pt]
\tikzstyle{uglabel}=[rounded corners=0.2em,fill=green!20,inner sep=0.1em,font=\scriptsize, anchor=west, xshift=-0.2em, yshift=0,opacity=1]
\tikzstyle{box-no-outline}=[rectangle, draw=white, fill=white, inner sep=2pt]
\tikzstyle{polynomial}=[regular polygon, regular polygon sides=3, shape border rotate= 180, fill=white,draw=black,inner sep = -10pt, rounded corners, rounded corners=10pt, minimum width=25pt, font=\scriptsize]

\pgfdeclarelayer{edgelayer}
\pgfdeclarelayer{nodelayer}
\pgfsetlayers{background,edgelayer,nodelayer,main}
\tikzstyle{none}=[inner sep=0mm]
\tikzstyle{every loop}=[]

\usepackage{etoolbox}

\newtoggle{extern}
\togglefalse{extern}

\newcommand{\tikzfig}[1]{
\input{./figures/#1.tikz}
}

\def\fig{}

\iftoggle{extern}{
	\usetikzlibrary{external}
	\tikzexternalize[prefix=tikzfigs/]
	
	\renewcommand{\tikzfig}[1]{
		\tikzsetnextfilename{#1}
		
\input{./figures/#1.tikz}
}
	
}{
	
}

\usetikzlibrary{circuits.ee.IEC}
\newcommand{\ground}
{\iftoggle{extern}{\tikzsetnextfilename{ground}}{}
\begin{tikzpicture}[circuit ee IEC,yscale=0.9,xscale=0.8]
\draw (0,1ex) to (0,0) node[ground,rotate=-90,xshift=.65ex] {};
\end{tikzpicture}}%

\newcommand{\sground}{\scalebox{0.5}{\ground}}

\newcommand{\rewrite}[1]{
\underset{\substack{#1}}{\longrightarrow}
}

\newcommand{\interp}[1]{\left\llbracket #1 \right\rrbracket}
\newcommand{\abs}[1]{\left\lvert #1 \right\rvert}

\newcommand{\bra}[1]{\ensuremath{\left\langle #1 \right|}}
\newcommand{\ket}[1]{\ensuremath{\left|  #1 \right\rangle}}
\newcommand{\ketbra}[2]{\ket{#1}\!\!\bra{#2}}

\newcommand{\cat}[1]{\mathbf{#1}}

\newcommand{\Var}{\operatorname{Var}}

\allowdisplaybreaks

\title{The Structure of Sum-Over-Paths, its Consequences, and Completeness for Clifford}

\author{
Renaud Vilmart
\institute{Universit\'e Paris-Saclay, CNRS, Laboratoire de Recherche en Informatique, 91405, Orsay, France}
\\
\email{vilmart@lri.fr}
}
\begin{document}

\maketitle

\begin{abstract}
We show that the formalism of ``Sum-Over-Path'' (SOP), used for symbolically representing linear maps or quantum operators, together with a proper rewrite system, has a structure of dagger-compact PROP. Several consequences arise from this observation:\\
 -- Morphisms of SOP are very close to the diagrams of the graphical calculus called ZH-Calculus, so we give a system of interpretation between the two\\
 -- A construction, called the discard construction, can be applied to enrich the formalism so that, in particular, it can represent the quantum measurement.

We also enrich the rewrite system so as to get the completeness of the Clifford fragments of both the initial formalism and its enriched version.
\end{abstract}

\section{Introduction}

The ``Sum-Over-Paths'' (SOP) formalism \cite{SOP} was introduced in order to perform verification on quantum circuits. It is inspired by Feynman's notion of path-integral, and can be conceived as a discrete version of it.

The core idea here is to represent unitary transformations in a symbolic way, so as to be able to simplify the term, which would for instance accelerate its evaluation. To do so, the formalism comes equipped with a rewrite system, which reduces any term into an equivalent one.

As pure quantum circuits (which represent unitary maps) can easily be mapped to an SOP morphism, one can try and perform verification: given a specification $\mathcal S$ and another SOP morphism $t$ obtained from a circuit supposed to verify the specification, we can compute the term $\mathcal S\circ t^{\dagger}$ and try to reduce it to the identity. In a very similar way, one can check whether two quantum circuits perform the same unitary map.

The rewrite system is known to be complete for Clifford unitary maps, i.e.~in the Clifford fragment of quantum mechanics, the term obtained from $t_1\circ t_2^{\dagger}$ will reduce to the identity iff $t_1$ and $t_2$ represent the same unitary map. Moreover, this reduction terminates in time polynomial in the size of the SOP term (itself related to the size of the quantum circuit), and still performs well outside the Clifford fragment.

Another use for this formalism is quantum simulation, the problem of evaluating the unitary map represented by a quantum circuit. Doing this is exponential in the number of variables in the SOP term, but the rewrite strategy reduces this number of variables, so each step in the reduction roughly divides the evaluation time by two.

Something that the SOP formalism cannot do for now however is circuit simplification. Indeed, even though we can easily translate an arbitrary quantum circuit to an SOP term, and then reduce it, there is no known way to extract a quantum circuit from the result.

We show in this paper that the formalism, when considered as a category (denoted $\cat{SOP}$), has the structure of a $\dagger$-compact PROP. This structure is explained in details in Section \ref{sec:props}. This structure is shared by a much larger set of maps than just unitary maps, namely $\cat{Qubit}$, the category whose morphisms are linear maps of $\mathbb C^{2^m}\times\mathbb C^{2^n}$. In particular, we show that any morphism of $\cat{Qubit}$ could be expressed as a morphism of $\cat{SOP}$.

Because the formalism is no longer restricted to unitary maps, we argue that it could benefit from a slight redefinition, which is done in Section \ref{sec:sop-v2}.

Another ``family'' of categories that share this structure is the family of graphical languages for quantum computation: ZX-Calculus, ZW-Calculus and ZH-Calculus \cite{ZH,interacting,ghz-w}. All three formalisms represent morphisms of $\cat{Qubit}$ using diagrams, and come with equational theories, proven to be complete for the whole category \cite{ZH,HNW,euler-zx}, i.e.~whenever two diagrams represent the same morphism of $\cat{Qubit}$, the first can be turned into the other using only the equational theory.

In Section \ref{sec:sop-zh}, we show that any diagram of the ZH-Calculus can be interpreted as a morphism of $\cat{SOP}$, and conversely, that any morphism of $\cat{SOP}$ can be turned into an equivalent ZH-diagram.

In Section \ref{sec:sop-clifford}, we realise that the rewrite system of $\cat{SOP}$ is not enough for the completeness of the Clifford fragment of $\cat{Qubit}$. We define a restriction of $\cat{SOP}$ that captures exactly this fragment, and enrich the set of rules so as to get the completeness in this restriction.

In Section \ref{sec:sop-discard}, we enrich the whole formalism using the discard construction \cite{CJPV19}, so as to be able to represent completely positive maps, as well as the operator of partial trace. Again, one can consider the Clifford fragment of this formalism. We give a new set of rewrite rules, and show that it makes the fragment complete.

\section{Background}
\label{sec:props}

\subsection{PROPs and String Diagrams}

The first kind of category we will be interested in is the \emph{PROP} \cite{Lack-PROP,PhD.Zanasi}. A PROP $\cat C$ is a strict symmetric monoidal category \cite{mac2013categories,selinger2010survey} generated by a single object, or equivalently, whose objects form $\mathbb N$. Hence the morphisms of $\cat C$ are of the form $f:n\to m$. They can be composed sequentially $(.\circ.)$ or in parallel $(.\otimes.)$, and they satisfy the following axioms:
\begin{align*}
f\circ (g\circ h) = (f\circ g)\circ h \quad&\qquad f\otimes (g\otimes h) = (f\otimes g)\otimes h\\
id_m \circ f = f = f\circ id_n \quad & \qquad id_0\otimes f = f = f\otimes id_0\\
(f_2\circ f_1)\otimes (g_2\circ g_1) &= (f_2\circ g_2)\circ(f_1\otimes g_1)
\end{align*}
The category is also equipped with a particular family of morphisms $\sigma_{n,m}:n+m\to m+n$. Intuitively, these allow morphisms to swap places. They satisfy additional axioms:
\begin{align*}
\sigma_{n,m+p} = (id_m\otimes \sigma_{n,p})\circ(\sigma_{n,m}\otimes id_p) \quad&\qquad \sigma_{n+m,p} = (\sigma_{n,p}\otimes id_m)\circ(id_n\otimes \sigma_{m,p})\\
\sigma_{m,n}\circ\sigma_{n,m} = id_{n+m} \quad&\qquad (id_p \otimes f)\circ \sigma_{n,p} = \sigma_{m,p}\circ(f\otimes id_p)
\end{align*}

Monoidal categories, and subsequently PROPs, have the benefit of having a nice graphical representation, using string diagrams. The object $n$ and equivalently $id_n$ is represented by $n$ parallel wires: \tikzfig{identity-n}; and a morphism $f:n\to m$ as a box with $n$ input wires and $m$ output wires: \mbox{\tikzfig{morphism-f}.}
The sequential composition $(.\circ.)$ is obtained by plugging the outputs of the morphism on the right to the inputs of the morphism of the left. The parallel composition $(.\otimes.)$ is obtained by putting the two diagrams side by side.

The first set of axioms is for coherence: the two compositions are associative, so we can forget about the parentheses, and the following string diagram is well defined, as:
$$\tikzfig{bifunctorial-law}$$
The morphism $\sigma_{n,m}$ is represented by \tikzfig{sigma}. The following axioms are satisfied:
$$\tikzfig{sigma-axiom}\qquad\qquad\tikzfig{sigma-axiom-bis}$$
$$\tikzfig{sigma-axiom-id}\qquad\qquad\tikzfig{sigma-axiom-2}$$

\subsection{$\dagger$-Compact PROPs}

Some PROPs can have additional structure, such as a compact-closed structure, or having a $\dagger$-functor.

A $\dagger$-PROP $\cat C$ is a PROP together with an involutive, identity-on-objects functor $(.)^\dagger: \cat C^{\operatorname{op}}\to \cat C$ compatible with $(.\otimes.)$. That is, for every morphism $f:n\to m$, there is a morphism $f^{\dagger}:m\to n$ such that $f^{\dagger\dagger} = f$. It behaves with the compositions by $(f\circ g)^{\dagger} = g^{\dagger}\circ f^{\dagger}$ and $(f\otimes g)^{\dagger}=f^{\dagger}\otimes g^{\dagger}$. Finally, we have $\sigma_{n,m}^{\dagger} = \sigma_{m,n}$.

A $\dagger$-compact PROP as two particular families of morphisms: $\eta_n:0\to 2n$ and $\epsilon_n:2n\to 0$. These are dual by the $\dagger$-functor: $\eta_n^\dagger = \epsilon_n$. They satisfy the following axioms:
\begin{align*}
(\epsilon_n\otimes id_n)\circ(id_n\otimes \eta_n) = id_n = (id_n\otimes\epsilon_n)\circ(\eta_n\otimes id_n) \\
\sigma_{n,n}\circ\eta_n = \eta_n \quad\qquad
\eta_{n+m} = (id_n \otimes \sigma_{n,m} \otimes id_m) \circ (\eta_n\otimes \eta_m)
\end{align*}
The morphisms $\eta_n$ and $\epsilon_n$ may be denoted \tikzfig{eta} and \tikzfig{epsilon}. They hence satisfy:
$$\tikzfig{snake}\qquad\qquad\tikzfig{snake-2}$$
$$\tikzfig{eta-axiom}$$

In this context, one can define the transpose operator of a morphism $f$ as:
$$f^t := (\epsilon_m \otimes id_n)\circ(id_m\otimes f\otimes id_n)\circ (id_m\otimes \eta_m)\qquad\text{i.e.}\quad\tikzfig{transpose}$$

One can check that, thanks to the axioms of $\dagger$-compact PROP, $(f\circ g)^t = g^t\circ f^t$, $(f\otimes g)^t = f^t\otimes g^t$, and $f^{tt} = f$.

We can then compose $(.)^t$ and $(.)^\dagger$: $\overline{(.)}:=(.)^{\dagger t}$. Again using the axioms of $\dagger$-compact PROP, one can check that $(.)^{\dagger t} = (.)^{t\dagger}$.

\subsection{Example: $\cat{Qubit}$}

The usual example of a strict symmetric $\dagger$-compact monoidal category is $\cat{FHilb}$, the category whose objects are finite dimensional Hilbert spaces, and whose morphisms are linear maps between them. It is not, however, a PROP, as it is not generated by a single object.

One subcategory of $\cat{FHilb}$ that \emph{is} a PROP, though, is $\cat{Qubit}$ the subcategory of $\cat{FHilb}$ generated by the object $\mathbb C^2$, considered as the object $1$. A morphism $f:n\to m$ of $\cat{Qubit}$ is hence a linear map from $\mathbb C^{2^n}$ to $\mathbb C^{2^m}$. $(.\circ.)$ is then the usual composition of linear maps, and $(.\otimes.)$ is the usual tensor product of linear maps. One can check that the first set of axioms is satisfied.

This is not enough to conclude that $\cat{Qubit}$ is a PROP. We still need to define a family of morphisms $\sigma_{n,m}$. In the Dirac notation, given a basis $\mathcal B$ of $\mathbb C^2$, we can define $\sigma_{n,m}$ as $\sigma_{n,m}:=\sum\limits_{(\vec x,\vec y)\in\mathcal B^n\times\mathcal B^m} \ketbra{\vec y,\vec x}{\vec x,\vec y}$. One can then check that all the axioms of PROPs are satisfied.

$\cat{Qubit}$ is not only a PROP, but also $\dagger$-compact. Indeed, first, given a morphism:
$$f = \sum\limits_{(\vec x,\vec y)\in\mathcal B^n\times\mathcal B^m} a_{\vec x,\vec y}\ketbra {\vec y}{\vec x}$$
we can define its dagger $f^\dagger:= \sum\limits_{(\vec x,\vec y)\in\mathcal B^n\times\mathcal B^m} \overline{a_{\vec x,\vec y}}\ketbra {\vec x}{\vec y}$, which is the usual definition of the dagger for linear maps.

Its compact structure can be given by $\eta_n := \sum\limits_{\vec x \in B^n} \ket{\vec x,\vec x}$, which implies $\epsilon_n = \eta_n^\dagger = \sum\limits_{\vec x \in B^n} \bra{\vec x,\vec x}$. One can check that all the axioms of $\dagger$-compact PROPs are satisfied.

Since $\cat{Qubit}$ is $\dagger$-compact, we can define the transpose $(.)^t$ which happens to be the usual transpose of linear maps, and the conjugate $\overline{(.)}$, which again is the usual conjugation in linear maps over $\mathbb C$.

There is a subcategory of $\cat{Qubit}$ that is of importance: $\cat{Stab}$. It is the smallest $\dagger$-compact subcategory of $\cat{Qubit}$ (the compact structure is preserved) that contains:
\begin{itemize}
\item $\ket0:0\to1$
\item $H:= \frac1{\sqrt2}(\ketbra00+\ketbra01+\ketbra10-\ketbra11):1\to1$
\item $S:= \ketbra00 + i\ketbra11:1\to1$
\item $CZ:= \ketbra{00}{00}+\ketbra{01}{01}+\ketbra{10}{10}-\ketbra{11}{11}:2\to2$
\end{itemize}

\section{The Category $\cat{SOP}$}
\label{sec:category-sop}

\subsection{$\cat{SOP}$ as a PROP}

The point of the Sum-Over-Paths formalism \cite{SOP}, is to \emph{symbolically} manipulate morphisms written in a form akin to the Dirac notation. Reasoning on symbolic terms allow us to detect where a term can be simplified in a ``smaller'' one, or to give a specification on a term.

A morphism of the category will be of the form $\ket{\vec x}\mapsto s\sum\limits_{\vec y\in V^k} e^{2i\pi P(\vec x,\vec y)}\ket{\vec Q(\vec x,\vec y)}$ where:
\begin{itemize}
\item $\vec x=x_1,\ldots,x_n$ is the input signature, it is a list of variables
\item $V$ is a set of variables (hence $\vec y$ is a collection of these variables)
\item $P$ is a multivariate polynomial, instantiated by the variables $\vec x$ and $\vec y$
\item $\vec Q = Q_1,\ldots,Q_m$ is the output signature, it is a multivariate, multivalued boolean polynomial
\item $s$ is a real scalar
\end{itemize}
We may denote $V_f$ a subset of the variables $V$ used in $f$. Then by default, if $V_f$ and $V_g$ are used in the same term, we consider that $V_f\cap V_g=\varnothing$. To distinguish the two sum operators (the one in $P$ and the one in $\vec Q$), we can denote the one in the output signature $\vec Q$ as $\oplus$. Moreover, it will sometimes be necessary to immerse one of the boolean polynomials $Q_i$ in the polynomial $P$. We hence define $\widehat{Q_i}$ inductively as $\widehat{x}=x$ for a variable $x$, $\widehat{pq}=\widehat{p}\widehat{q}$ and $\widehat{p\oplus q} = \widehat{p}+\widehat{q}-2\widehat{pq}$.


\begin{definition}[$\cat{SOP}$]
$\cat{SOP}$ is defined as the PROP where, given a set of variables $V$:
\begin{itemize}
\item Identity morphisms are $id_n : \ket{\vec x}\mapsto \ket{\vec x}$
\item Morphisms $f:n\to m$ are of the form $f:\ket{\vec x}\mapsto s\sum\limits_{\vec y\in V^k} e^{2i\pi P(\vec x,\vec y)}\ket{\vec Q(\vec x,\vec y)}$ where $s\in\mathbb{R}$, $\vec x\in V^n$, $P\in \mathbb R[X_1,\ldots,X_{n+k}]/(1,X_i^2-X_i)$, and $\vec Q \in \left(\mathbb F_2[X_1,\ldots,X_{n+k}]\right)^m$
\item Composition is obtained as $f\circ g := \ket{\vec x_g}\mapsto s_fs_g\sum\limits_{\substack{\vec y_f\in V_f^{k_f}\\\vec y_g\in V_g^{k_g}}} e^{2i\pi (P_g+P_f[\vec x_f\leftarrow \widehat{\vec Q_g}])}\ket{\vec Q_f[\vec x_f\leftarrow {\vec Q_g}]}$
\item Tensor product is obtained as $f\otimes g := \ket{\vec x_f\vec x_g}\mapsto s_fs_g\sum\limits_{\substack{\vec y_f\in V_f^{k_f}\\\vec y_g\in V_g^{k_g}}} e^{2i\pi (P_g+P_f)}\ket{\vec Q_f\vec Q_g}$
\item The symmetric braiding is $\sigma_{n,m}: \ket{\vec x_1, \vec x_2}\mapsto \ket{\vec x_2,\vec x_1}$
\end{itemize}
\end{definition}

The polynomial $P$ is called the \emph{phase polynomial}, as it appears in the morphism in $e^{2i\pi.}$. Because of this, we consider the polynomial modulo $1$. We also consider the polynomial quotiented by $X^2-X$ for all its variables $X$, as these variables are to be evaluated in $\{0,1\}$, so we consider $X^2=X$.

Notice that the definition of the identities does not directly fit the description of the morphisms. However, we can rewrite it as $\ket{\vec x}\mapsto \ket{\vec x} = \ket{\vec x}\mapsto 1\sum\limits_{y \in V^0} e^{2i\pi 0}\ket{\vec x}$. Hence, when we sum over a single element, we may forget the sum operator, and when the phase polynomial is 0, we may not write it. Notice by the way that $id_0 = \ket{}\mapsto \ket{}$. Indeed, $\ket{}$ is absolutely valid, it represents an empty register.

\begin{example}
We can give the $\cat{SOP}$ version of the usual quantum gates:\\
\begin{minipage}{0.45\columnwidth}
\begin{align*}
R_Z(\alpha)&:= \ket{x}\mapsto e^{2i\pi\frac{\alpha x}{2\pi}}\ket x\\
H &:= \ket x \mapsto \frac{1}{\sqrt{2}}\sum_{y\in V} e^{2i\pi \frac{xy}{2}}\ket y
\end{align*}
\end{minipage}
\hfill
\begin{minipage}{0.45\columnwidth}
\begin{align*}
\textit{CNot} &:= \ket{x_1,x_2}\mapsto \ket{x_1,x_1{\oplus} x_2}\\
CZ &:= \ket{x_1,x_2}\mapsto e^{2i\pi\frac{x_1x_2}{2}}\ket{x_1,x_2}
\end{align*}
\end{minipage}

\end{example}

\begin{example}
Let us derive the operation $(id\otimes H)\circ \textit{CNot}\circ (id\otimes H)$:
\begin{align*}
&(id\otimes H)\circ \textit{CNot}\circ (id\otimes H)\\
&= (id\otimes H)\circ\left(\vphantom{\rule{1pt}{2em}}\ket{x_1,x_2}\mapsto \ket{x_1,x_1{\oplus} x_2}\right)\circ\left(\ket{x_1,x_2} \mapsto \frac{1}{\sqrt{2}}\sum_{y\in V} e^{2i\pi \frac{x_2y}{2}}\ket{x_1,y}\right)\\
&=\left(\ket{x_1,x_2} \mapsto \frac{1}{\sqrt{2}}\sum_{y\in V} e^{2i\pi \frac{x_2y}{2}}\ket{x_1,y}\right)\circ\left(\ket{x_1,x_2} \mapsto \frac{1}{\sqrt{2}}\sum_{y_1\in V} e^{2i\pi \frac{x_2y_1}{2}}\ket{x_1,x_1{\oplus}y_1}\right)\\
&=\ket{x_1,x_2} \mapsto \frac12\sum_{y_1,y_2\in V} e^{2i\pi\left(\frac{x_2y_1}{2}+\frac{x_1+y_1-2x_1y_1}{2}y_2\right)}\ket{x_1,y_2}
\end{align*}
where $x_1+y_1-2x_1y_1=\widehat{x_1\oplus y_1}$.
\end{example}

The previous definition contains a claim: that $\cat{SOP}$ is a PROP. To be so, one has to check all the axioms of PROPs. One has to be careful when doing so. Indeed, the sequential composition $(.\circ.)$ induces a substitution. Hence, one has to check all the axioms in the presence of a ``context'', that is, one has to show that the axioms can be applied \emph{locally}.

If an axiom states \tikzfig{axiom-t}, one should ideally check that \tikzfig{axiom-t-context} for any ``before'' morphism $B$ and any ``after'' morphism $A$. However, this can be easily reduced to checking that \tikzfig{axiom-t-context-2}.

In the case of the axioms of PROPs, this can further be reduced to showing the axioms without context, as neither $id_n$ nor $\sigma_{n,m}$ introduce variables or phases. For the other axioms, however, the context will have to be taken into account.

A fairly straightforward but tedious verification gives that, indeed, $\cat{SOP}$ is a PROP. We give as an example the proof of associativity of the sequential composition (without context for simplicity):
\begin{align*}
(f\circ g)\circ h &= \left(\ket{\vec x}\mapsto s_gs_f\sum e^{2i\pi\left(P_g+P_f[\vec x_f\leftarrow \widehat{\vec Q_g}]\right)}\ket{\vec Q_f[\vec x_f\leftarrow \vec Q_g]}\right)\circ h \\
&= \ket{\vec x}\mapsto s_gs_fs_h\sum e^{2i\pi\left(P_h+P_g[\vec x_g\leftarrow \widehat{\vec Q_h}]+P_f[\vec x_f\leftarrow \widehat{\vec Q_g},\vec x_g\leftarrow \widehat{\vec Q_h}]\right)}\ket{\vec Q_f[\vec x_f\leftarrow \vec Q_g,\vec x_g\leftarrow \vec Q_h]}\\
&= \ket{\vec x}\mapsto s_gs_fs_h\sum e^{2i\pi\left(P_h+P_g[\vec x_g\leftarrow \widehat{\vec Q_h}]+P_f[\vec x_f\leftarrow \widehat{\vec Q_g}[\vec x_g\leftarrow \widehat{\vec Q_h}]]\right)}\ket{\vec Q_f[\vec x_f\leftarrow \vec Q_g[\vec x_g\leftarrow \vec Q_h]]}\\
&= f\circ \left(\ket{\vec x}\mapsto s_gs_h\sum e^{2i\pi\left(P_h+P_g[\vec x_g\leftarrow \widehat{\vec Q_h}]\right)}\ket{\vec Q_g[\vec x_g\leftarrow \vec Q_h]}\right) = f\circ(g\circ h)
\end{align*}
or that $\sigma$ swaps the places of morphisms:
\begin{align*}
(id_p \otimes f)\circ \sigma_{n,p} =& \left(\ket{\vec x_1,\vec x_2} \mapsto s\sum e^{2i\pi P_f}\ket{\vec x_1, \vec Q_f}\right)\circ\left(\ket{\vec x_1,\vec x_2}\mapsto\ket{\vec x_2,\vec x_1}\right)\\
&=\ket{\vec x_1,\vec x_2} \mapsto s\sum e^{2i\pi P_f}\ket{\vec Q_f,\vec x_1}\\
&=\left(\ket{\vec x_1,\vec x_2}\mapsto\ket{\vec x_2,\vec x_1}\right)\circ\left(\ket{\vec x_1,\vec x_2} \mapsto s\sum e^{2i\pi P_f}\ket{\vec Q_f,\vec x_2}\right)
= \sigma_{m,p}\circ(f\otimes id_p)
\end{align*}

\subsection{From $\cat{SOP}$ to $\cat{Qubit}$}

To check the soundness of what we are going to do in the following, it may be interesting to have a way of interpreting morphisms of $\cat{SOP}$ as morphisms of $\cat{Qubit}$.

\begin{definition}
The functor $\interp{.}:\cat{SOP}\to\cat{Qubit}$ is defined as being identity on objects, and such that, for $f=\ket{\vec x}\mapsto s\sum\limits_{\vec y\in V^k} e^{2i\pi P(\vec x,\vec y)}\ket{Q(\vec x,\vec y)}$:
$$\interp{f} := \hspace*{2em}s\hspace*{-2em}\sum\limits_{(\vec x,\vec y)\in \{0,1\}^n\times\{0,1\}^k}\hspace*{-2em} e^{2i\pi P(\vec x,\vec y)}\ketbra{Q(\vec x,\vec y)}{\vec x}$$
\end{definition}

\begin{example}
The interpretation of $H$ is as intended the Hadamard gate:
$$\interp{H} = \frac{1}{\sqrt{2}}\sum_{x,y\in \{0,1\}} e^{2i\pi \frac{xy}{2}}\ketbra y x = \frac{1}{\sqrt2}\left(\ketbra00+\ketbra01+\ketbra10-\ketbra11\right)$$
\end{example}

\begin{proposition}
The interpretation $\interp{.}$ is a \emph{PROP-functor}, meaning:
\begin{itemize}
\item $\interp{.\circ.}=\interp{.}\circ\interp{.}$
\item $\interp{.\otimes.}=\interp{.}\otimes\interp{.}$
\item $\interp{\sigma_{n,m}}=\sigma_{n,m}$
\end{itemize}
\end{proposition}

\begin{proof}
This is a straightforward verification.
\end{proof}

\subsection{$\cat{SOP}$ as a $\dagger$-Compact PROP}

\subsection{Towards a Compact Structure}

It is tempting to try and adapt the compact structure of $\cat{Qubit}$ to $\cat{SOP}$. To do so, we can first define $\eta_n:= \ket{}\mapsto\sum\limits_{\vec y\in V^n}\ket{\vec y,\vec y}$. However, we cannot as easily define $\epsilon_n$. What $\epsilon_1$ intuitively does in $\cat{Qubit}$ is: given two inputs $x_1$ and $x_2$, it checks if they are equal, if so it returns the scalar $1$, if not, the scalar $0$.

In $\cat{SOP}$ we can force two variables to be equal, using a third fresh variable $y$. Indeed, consider the sum $\sum e^{2i\pi(\frac{x_1+x_2}{2}y+P)}$ where $y$ is fresh i.e.~not used in $P$. Then, if the variables $x_1$ and $x_2$ are different, then $$\sum e^{2i\pi(\frac{x_1+x_2}{2}y+P)} = \sum e^{2i\pi(\frac{y}{2}+P)} = \sum e^{2i\pi(0+P)} + \sum e^{2i\pi(\frac{1}{2}+P)} = \sum e^{2i\pi P} - \sum e^{2i\pi P} = 0$$

Hence, we can define $\epsilon_1$ as $\epsilon_1:=\ket{x_1,x_2}\mapsto \frac12\sum\limits_{y\in V} e^{2i\pi \frac{x_1+x_2}{2}y} \ket{}$ and even extend it to arbitrary objects: $\epsilon_n := \ket{\vec x_1,\vec x_2}\mapsto\frac{1}{2^n}\sum\limits_{\vec y\in V^n} e^{2i\pi \frac{\vec x_1\cdot \vec y+\vec x_2\cdot\vec y}{2}} \ket{}$.

We can check that $\interp{\epsilon_1}=\epsilon_1$:
\begin{align*}
\interp{\epsilon_1} &= \frac12\sum_{x_i,y\in \{0,1\}} e^{2i\pi\frac{x_1+x_2}{2}y}\ketbra{}{x_1,x_2}
= \frac12\sum_{x_i\in \{0,1\}} (1+e^{i\pi (x_1+x_2)})\bra{x_1,x_2}\\
&=\bra{00}+\bra{11}
\end{align*}
Similarly, $\interp{\epsilon_n} = \epsilon_n$.

We can now try to check whether the axioms of $\dagger$-compact PROPs (at least the ones that do not require the $\dagger$, as we have not defined it yet) are satisfied:

\begin{align*}
\sigma_{n,n}\circ\eta_n = \left(\vphantom{\rule{1pt}{2em}}\ket{\vec x_1, \vec x_2}\mapsto \ket{\vec x_2, \vec x_1}\right)\circ\left(\ket{}\mapsto \sum_{\vec y\in V^n} \ket{\vec y,\vec y}\right)=\ket{}\mapsto \sum_{\vec y\in V^n} \ket{\vec y,\vec y}=\eta_n
\end{align*}

\begin{align*}
(id_n \otimes \sigma_{n,m} \otimes & id_m) \circ (\eta_n\otimes \eta_m)\\
&= \left(\vphantom{\rule{1pt}{2em}}\ket{\vec x_1, \vec x_2,\vec x_3, \vec x_4}\mapsto \ket{\vec x_1, \vec x_3,\vec x_2, \vec x_4}\right)\circ\left(\ket{}\mapsto \sum_{\vec y_1\in V^n,\vec y_2\in V^m} \ket{\vec y_1,\vec y_1, \vec y_2, \vec y_2}\right)\\
&=\ket{}\mapsto \sum_{\vec y_1\in V^n,\vec y_2\in V^m} \ket{\vec y_1,\vec y_2, \vec y_1, \vec y_2}=\ket{}\mapsto \sum_{\vec y\in V^{n+m}} \ket{\vec y,\vec y}
=\eta_{n+m}
\end{align*}
These two equations were shown without a context for simplicity, but still hold with it.

However, the equation:
$$(\epsilon_n\otimes id_n)\circ(id_n\otimes \eta_n) = id_n = (id_n\otimes\epsilon_n)\circ(\eta_n\otimes id_n)$$
is not satisfied, as:
\begin{align*}
(\epsilon_n\otimes id_n)\circ(id_n\otimes \eta_n) = \ket{\vec x}\mapsto \frac12 \sum_{\vec y_1,\vec y_2\in V^n} e^{2i\pi\frac{\vec x\cdot\vec y_2+ \vec y_1\cdot\vec y_2}{2}}\ket{\vec y_1}\neq id_n
\end{align*}

The fact that we have $(\epsilon_n\otimes id_n)\circ(id_n\otimes \eta_n) \neq id_n$ in $\cat{SOP}$ but $\interp{(\epsilon_n\otimes id_n)\circ(id_n\otimes \eta_n)} = \interp{id_n}$ in $\cat{Qubit}$ hints at a way to \emph{rewrite} the first term as the second in $\cat{SOP}$.

\subsection{An Equational Theory}

\begin{figure}[!htb]
\[\sum_{\vec y} e^{2i\pi P}\ket{\vec Q}
\rewrite{y_0\notin\Var(P,\vec Q)} 2\sum_{\vec y\setminus{\{y_0\}}} e^{2i\pi P}\ket{\vec Q}\tag{Elim}\]

\[\sum_{\vec y} e^{2i\pi\left(\frac{y_0}{2} (y_0' + \widehat{Q_2}) + R\right)}\ket{\vec Q}
\rewrite{y_0\notin\Var(R,Q_2,\vec Q)\\y_0'\notin \Var(Q_2)} 2\!\!\sum_{\vec y\setminus{\{y_0,y_0'\}}}\!\! e^{2i\pi\left(R\left[y_0'\leftarrow \widehat{Q_2}\right]\right)}\ket{\vec Q\left[y_0'\leftarrow Q_2\right]}\tag{HH}\]

\[\sum_{\vec y} e^{2i\pi\left(\frac{y_0}{4} + \frac{y_0}{2}\widehat{Q_2} + R\right)}\ket{\vec Q}
\rewrite{y_0\notin\Var(Q_2,R,\vec Q)} \sqrt{2}\sum_{\vec y\setminus{\{y_0\}}} e^{2i\pi\left(\frac{1}{8}-\frac{1}{4}\widehat{Q_2} + R\right)}\ket{\vec Q}\tag{$\omega$}\]
\caption[]{Rewrite strategy $\rewrite{\operatorname{Clif}}$.}
\label{fig:rewrite-rules-1}
\end{figure}

A rewrite strategy is given in \cite{SOP}, and we show in Figure \ref{fig:rewrite-rules-1} the ones we are going to use in the paper.
Each rewrite rule contains a condition, which usually ensures that a variable (the one we want to get rid of) does not appear in some polynomials. We hence use $\Var$ as the operator that gets all the variables from a sequence of polynomials:
$$\begin{cases}
\Var(Q_1,Q_2,\ldots) = \Var(Q_1)\cup\Var(Q_2)\cup\ldots\\
\Var(Q_1\oplus Q_2) = \Var(Q_1)\cup\Var(Q_2)\\
\Var(Q_1Q_2) = \Var(Q_1)\cup\Var(Q_2)\\
\Var(y) = \{y\} \text{ if $y\in V$}\\
\Var(0)=\Var(1)=\varnothing
\end{cases}$$
For simplicity, the input signature is omitted, as well as the parameters in the polynomials.

$\rewrite{\operatorname{Clif}}$ denotes the rewrite system formed by the three rules (Elim), (HH) and ($\omega$). $\overset\ast{\rewrite{\operatorname{Clif}}}$ is the transitive closure of the rewrite system. Notice that all the rules remove at least one variable from the morphism, so we know $\rewrite{\operatorname{Clif}}$ terminates.

When the rules are not oriented, we get an equivalence relation on the morphisms of $\cat{SOP}$. We denote this equivalence $\underset{\operatorname{Clif}}\sim$.

We denote $\cat{SOP}/\underset{\operatorname{Clif}}\sim$ the category $\cat{SOP}$ quotiented by the equivalence relation $\underset{\operatorname{Clif}}\sim$. This newly obtained category is still a PROP. It even has a compact structure, as the last necessary axiom is now derivable:
\begin{align*}
A\circ(\epsilon\otimes id)\circ(id\otimes \eta)\circ B = \ket{ x}\mapsto \frac{s_As_B}2 \sum_{\substack{\vec y_A,\vec y_B\\ y_1, y_2\in V}} e^{2i\pi (P_B+(\widehat{Q_B}+y_1)\frac{y_2}{2}+P_A[x\leftarrow y_1])}\ket{Q_A[x\leftarrow y_1]}\\
\rewrite{\text{(HH)}}\ket{ x}\mapsto {s_As_B} \sum_{\vec y_A,\vec y_B} e^{2i\pi (P_B+P_A[x\leftarrow y_1][y_1\leftarrow \widehat{Q_B}])}\ket{Q_A[x\leftarrow y_1][y_1\leftarrow \widehat{Q_B}]}\\
=\ket{ x}\mapsto {s_As_B} \sum_{\vec y_A,\vec y_B} e^{2i\pi (P_B+P_A[x\leftarrow \widehat{Q_B}])}\ket{Q_A[x\leftarrow \widehat{Q_B}]} = A\circ B
\end{align*}

\subsection{$\dagger$-Functor for $\cat{SOP}$}

To show that $\cat{SOP}/\underset{\operatorname{Clif}}\sim$ is $\dagger$-compact, we lack a notion of $\dagger$-functor $\cat{SOP}$.

Remember that we defined $\overline{(.)}$ as $(.)^{\dagger t}$. Since we have a compact structure, we can already define the functor $(.)^t$. Thanks to the new equivalence relation $\underset{\operatorname{Clif}}\sim$, this functor is involutive. Hence, we have $(.)^\dagger = \overline{(.)}^t$. An appropriate definition of the conjugation can be given:

\begin{definition}
The conjugation is defined on morphisms $f = \ket{\vec x}\mapsto s_f\sum e^{2i\pi P_f}\ket{\vec Q_f}$ as:
$$\overline{f}:= \ket{\vec x}\mapsto {s_f}\sum e^{-2i\pi P_f}\ket{\vec Q_f}$$
\end{definition}

This gives a definition of the $\dagger$. For the record, if $f$ is of the above form:
\begin{align*}
f^t &= \ket{\vec x}\mapsto \frac{s_f}{2^m}\sum e^{2i\pi \left(P_f+\frac{\vec{\widehat{Q_f}}[\vec x_f \leftarrow \vec y]\cdot \vec y'+\vec x\cdot\vec y'}2\right)}\ket{\vec y}\\
f^\dagger &= \ket{\vec x}\mapsto \frac{s_f}{2^m}\sum e^{2i\pi \left(-P_f+\frac{\vec{\widehat{Q_f}}[\vec x_f \leftarrow \vec y]\cdot \vec y'+\vec x\cdot\vec y'}2\right)}\ket{\vec y}
\end{align*}

These three functors are the expected ones:

\begin{proposition}
\label{prop:transpose-conjugate-dagger}
$\interp{(.)^t} = \interp{.}^t\ ,\quad
\interp{\overline{(.)}} = \overline{\interp{.}}\ ,\quad
\interp{(.)^\dagger} = \interp{.}^\dagger$
\end{proposition}

\begin{proof}
In appendix at page \pageref{prf:transpose-conjugate-dagger}.
\end{proof}

We can finally prove the wanted result:

\begin{theorem}
\label{thm:dagger-compact-prop}
$\cat{SOP}/\underset{\operatorname{Clif}}\sim$ is a $\dagger$-compact PROP.
\end{theorem}

\begin{proof}
In appendix, at page \pageref{prf:dagger-compact-prop}.
\end{proof}

\section{Redefinition of $\cat{SOP}$}
\label{sec:sop-v2}

In $\cat{Qubit}$, and hence in $\cat{SOP}$, it may feel unnatural to have asymmetrical input and outputs. Why not have morphisms of the form $f= s\sum_{\vec y} e^{2i\pi P}\ketbra{\vec O}{\vec I}$? In this case, we have to change the definition of the composition, and, because of this, the $\cat{SOP}$ morphisms do not form a category. However, it is a category when quotiented by $\underset{\operatorname{Clif}}\sim$. This is the reason why we did not define $\cat{SOP}$ like this at first, although it greatly simplifies the notions of compact structure and $\dagger$-functor.

We now redefine $\cat{SOP}$, and will use this new definition in the rest of the paper:
\begin{definition}[$\cat{SOP}$]
We redefine $\cat{SOP}$ as the collection of objects $\mathbb N$ and morphisms between them:
\begin{itemize}
\item Identity morphisms are $id_n : \sum\limits_{\vec y\in V^n}\ketbra{\vec y}{\vec y}$
\item Morphisms $f:n\to m$ are of the form $f: s\sum\limits_{\vec y\in V^k} e^{2i\pi P(\vec y)}\ketbra{\vec O(\vec y)}{\vec I(\vec y)}$ where $s\in\mathbb{R}$, $P\in \mathbb R[X_1,\ldots,X_{k}]/(1,X_i^2-X_i)$, $\vec O \in \left(\mathbb F_2[X_1,\ldots,X_{k}]\right)^m$ and $\vec I \in \left(\mathbb F_2[X_1,\ldots,X_{k}]\right)^n$
\item Composition is obtained as $f\circ g := \frac{s_fs_g}{2^{\abs{\vec I_f}}}\sum\limits_{\substack{\vec y_f,\vec y_g\\\vec y\in V^m}} e^{2i\pi \left(P_g+P_f+\frac{\vec O_g\cdot \vec y+\vec I_f\cdot \vec y}{2}\right)}\ketbra{\vec O_f}{\vec I_g}$
\item Tensor product is obtained as $f\otimes g := s_fs_g\sum\limits_{\substack{\vec y_f,\vec y_g}} e^{2i\pi (P_g+P_f)}\ketbra{\vec O_f\vec O_g}{\vec I_f\vec I_g}$
\item The symmetric braiding is $\sigma_{n,m}= \sum\limits_{\vec y_1,\vec y_2}\ketbra{\vec y_2,\vec y_1}{\vec y_1,\vec y_2}$
\item The compact structure is $\eta_n= \sum\limits_{\vec y} \ketbra{\vec y,\vec y}{}$ and $\epsilon_n= \sum\limits_{\vec y} \ketbra{}{\vec y,\vec y}$
\item The $\dagger$-functor is given by: $f^\dagger :=  s\sum\limits_{\vec y} e^{-2i\pi P}\ketbra{\vec I}{\vec O}$
\item The functor $\interp{.}$ is defined as: $\interp{f}:= s\sum\limits_{\vec y\in \{0,1\}^k} e^{2i\pi P(\vec y)}\ketbra{\vec O(\vec y)}{\vec I(\vec y)}$
\end{itemize}
\end{definition}

As announced, this is not a category, as $id\circ id = \frac12\sum_{\vec y} e^{2i\pi \frac{y_1+y_2}{2}y_3}\ketbra{y_2}{y_1}\neq \sum_y \ketbra yy = id$. This problem is solved by reintroducing the rewrite rules, that we adapt to the new formalism in Figure~\ref{fig:rewrite-rules-1-bis}.
\begin{figure}[!htb]
\[\sum_{\vec y} e^{2i\pi P}\ketbra{\vec O}{\vec I}
\rewrite{y_0\notin\Var(P,\vec O,\vec I)} 2\sum_{\vec y\setminus{\{y_0\}}} e^{2i\pi P}\ketbra{\vec O}{\vec I}\tag{Elim}\]

\[\sum_{\vec y} e^{2i\pi\left(\frac{y_0}{2} (y_0' + \widehat{Q}) + R\right)}\ketbra{\vec O}{\vec I}
\rewrite{y_0\notin\Var(R,Q,\vec O,\vec I)\\y_0'\notin \Var(Q)} 2\!\!\sum_{\vec y\setminus{\{y_0,y_0'\}}}\!\! e^{2i\pi\left(R\left[y_0'\leftarrow \widehat{Q}\right]\right)}\left(\ketbra{\vec O}{\vec I}\right)\left[y_0'\leftarrow Q\right]\tag{HH}\]

\[\sum_{\vec y} e^{2i\pi\left(\frac{y_0}{4} + \frac{y_0}{2}\widehat{Q} + R\right)}\ketbra{\vec O}{\vec I}
\rewrite{y_0\notin\Var(Q,R,\vec O,\vec I)} \sqrt{2}\sum_{\vec y\setminus{\{y_0\}}} e^{2i\pi\left(\frac{1}{8}-\frac{1}{4}\widehat{Q} + R\right)}\ketbra{\vec O}{\vec I}\tag{$\omega$}\]
\caption[]{Rewrite strategy $\rewrite{\operatorname{Clif}}$.}
\label{fig:rewrite-rules-1-bis}
\end{figure}

Again, we give the same name to the rewrite system, but this last one is the one we will use in the rest of the paper.

\begin{proposition}
$\cat{SOP}/\underset{\operatorname{Clif}}\sim$ is a $\dagger$-compact PROP, and $\interp{.}$ is a $\dagger$-compact PROP-functor.
\end{proposition}

\begin{remark}
In this new formalism, it is fairly easy to perform \emph{weak simulation}: given a quantum circuit $\mathcal C$ and two quantum states $\ket{\psi_i}$ and $\ket{\psi_o}$, what is the probability of outputting $\ket{\psi_o}$ when the circuit $\mathcal{C}$ is applied to the input $\ket{\psi_i}$?

Given $\cat{SOP}$-morphisms $t_{\mathcal C}$ for the circuit and $t_i$ and $t_o$ for the states $\ket{\psi_i}$ and $\ket{\psi_o}$, one simply needs to compute $t_o^{\dagger}\circ t_{\mathcal C}\circ t_i$ and simplify the term (which represents a scalar), before evaluating it.

Obviously, the efficiency of this method is conditioned by the simplification strategy used before evaluation.
\end{remark}

\begin{remark}
When building a $\cat{SOP}$ morphism $t$ from a circuit (or a diagram as we will show in the following) in this formalism, the resulting $t$ is always of size $O(d\times n)$ where $n$ is the size of the register, and $d$ the \emph{depth} of the circuit (and for a diagram in $O(G\times a)$ where $G$ is the number of generators and $a$ the maximum arity of these generators).

This contrasts with the first definition of $\cat{SOP}$, where the size of the constructed $\cat{SOP}$ term is polynomial for fixed restrictions of quantum mechanics (where the gates $R_Z$ are restricted to parameters that are multiples of $\frac\pi{2^{p-1}}$ for a fixed $p$), but exponential in general. This is due to the substitution $[x\leftarrow \widehat Q]$ in the composition. Indeed, if $Q$ contains $\ell$ monomials, $\widehat Q$ contains in the worst case $2^\ell-1$ monomials. In the considered fragment, the size is constrained as $\frac1{2^p}\widehat Q \bmod 1$ has at most $\sum\limits_{k=1}^p \begin{pmatrix}\ell\\k\end{pmatrix}\leq p\ell^p$ monomials.
\end{remark}

\section{$\cat{SOP}$ and Graphical Languages}
\label{sec:sop-zh}

The sum-over-paths formalism was initially intended to be used for isometries. As such, it was given a weak form of completeness -- as we will discuss in the next section. However, if transforming a quantum circuit -- that describes an isometry -- into an $\cat{SOP}$ morphism is easy, the converse, transforming a $\cat{SOP}$ morphism into a circuit is not. And actually, all $\cat{SOP}$ morphisms do not represent an isometry. For instance, the morphism $\epsilon_1$ described above is not an isometry. An even smaller example is $\sum_y\ketbra{}y$ 
which is a valid $\cat{SOP}$ morphism, but clearly does not represent an isometry.

The fact that $\cat{SOP}$ is $\dagger$-compact hints at another (family) of language(s) more suited for representing it: the Z$\ast$-Calculi: ZX, ZW and ZH. These are all $\dagger$-compact graphical languages, that have an interpretation in $\cat{Qubit}$, and are universal for $\cat{Qubit}$. This means that any morphism of $\cat{Qubit}$ can be represented as a morphism of either of these 3 languages.

The language that happens to be the closest to $\cat{SOP}$ is the ZH-Calculus. This is the one we are going to present in the following. However, bear in mind that, as we have semantics-preserving functors between any two of these three languages, one can do the same work with ZX and ZW-Calculi.

\subsection{The ZH-Calculus}


$\cat{ZH}$ is a PROP whose morphisms are composed (sequentially $(.\circ.)$ or in parallel $(.\otimes.)$) from Z-spiders and H-spiders:
\begin{itemize}
\item $Z_m^n:n\to m::$\tikzfig{Z-spider}, called Z-spider
\item $H_m^n(r):n\to m::$\tikzfig{H-spider}, called H-spider, with a parameter $r\in\mathbb C$
\end{itemize}
When $r$ is not specified, the parameter in the H-spider is taken to be $-1$.

$\cat{ZH}$ is made a $\dagger$-compact PROP, which means it also has the symmetric structure $\sigma$, the compact structure ($\eta$, $\epsilon$), and a $\dagger$-functor $(.)^\dagger:\cat{ZH}^{\operatorname{op}}\to\cat{ZH}$. It is defined by:\\
$$(Z_m^n)^\dagger := Z_n^m \qquad\text{and}\qquad(H_m^n(r))^\dagger := H_n^m(\overline{r})$$

For convenience, we define two additional spiders:\\
\tikzfig{X-spider} and \tikzfig{X-spider-neg}

The language comes with a way of interpreting the morphisms as morphisms of $\cat{Qubit}$. The standard interpretation $\interp{.}:\cat{ZH}\to\cat{Qubit}$ is a $\dagger$-compact-PROP-functor, defined as:
$$\interp{\tikzfig{Z-spider}} = \ketbra{0^m}{0^n} + \ketbra{1^m}{1^n}$$
$$\interp{\tikzfig{H-spider}} = \sum_{j_k,i_k\in\{0,1\}}r^{j_1\ldots j_mi_1\ldots i_n}\ketbra{j_1,\ldots,j_m}{i_1,\ldots,i_n}$$
Notice that we used the same symbol for two different functors: the two interpretations $\interp{.}:\cat{SOP}\to\cat{Qubit}$ and $\interp{.}:\cat{ZH}\to\cat{Qubit}$. It should be clear from the context which one is to be used.

The language is universal: $\forall f\in \cat{Qubit},~\exists D_f\in\cat{ZH},~~\interp{D_f} = f$. In other words, the interpretation $\interp{.}$ is onto. This is shown using a notion of normal form. This means there is a functor $\mathcal N:\cat{Qubit}\to\cat{ZH}$.

The language comes with an equational theory, which in particular gives the axioms for a $\dagger$-compact PROP. We will not present it here.

\subsection{From $\cat{ZH}$ to $\cat{SOP}$}

We show in this section how any $\cat{ZH}$ morphism can be turned into a $\cat{SOP}$ morphism in a way that preserves the semantics. We define $\left[.\right]^{\operatorname{sop}}:\cat{ZH}\to\cat{SOP}$ as the $\dagger$-compact PROP-functor such that:
$$\left[~\tikzfig{H-spider-0}~\right]^{\operatorname{sop}} := \sum \ketbra{y_1,\ldots, y_m}{x_1,\ldots,x_n}$$
$$\left[~\tikzfig{H-spider-phase}~\right]^{\operatorname{sop}} := \sum e^{2i\pi \frac{\alpha}{2\pi}x_1\ldots x_ny_1\ldots y_m}\ketbra{y_1,\ldots, y_m}{x_1,\ldots,x_n}$$
$$\left[~\tikzfig{H-scalar}~\right]^{\operatorname{sop}} := s\ketbra{}{} \quad\text{ for }s\in\mathbb R$$
$$\left[~\tikzfig{Z-spider}~\right]^{\operatorname{sop}} := \sum_y \ketbra{y,\ldots, y}{y,\ldots, y}$$
This does not give a full description of $[.]^{\operatorname{sop}}$, as we did not describe the interpretation of the H-spider for all parameters, but only for phases and $0$. However, any H-spider can be decomposed using the previous ones:

\begin{lemma}
\label{lem:H-spider-decomp}
For any $r\in\mathbb C$ such that $\abs{r}\notin\{0,1\}$, there exist $s\in\mathbb C$, $\alpha,\beta\in\mathbb R$ such that:
$$\interp{\tikzfig{H-spider}} = \interp{\tikzfig{H-spider-decomp}}$$
\end{lemma}

\begin{proof}
In appendix at page \pageref{prf:H-spider-decomp}.
\end{proof}

As a consequence, we extend the definition of $[.]^{\operatorname{sop}}$ by:
$$\left[~\tikzfig{H-spider}~\right]^{\operatorname{sop}} :=\left[~\tikzfig{H-spider-decomp}~\right]^{\operatorname{sop}} $$

This interpretation of ZH-diagrams as $\cat{SOP}$-morphisms preserves the semantics:
\begin{proposition}
\label{prop:sop-preserves-semantics}
$\interp{[.]^{\operatorname{sop}}} = \interp{.}$. In other words, the following diagram commutes:
$$\tikzfig{sop-interp-cd}$$
\end{proposition}

\begin{proof}
This is a straightforward verification.
\end{proof}

\subsection{From $\cat{SOP}$ to $\cat{ZH}$}

As we explained previously, there exists a functor from $\cat{Qubit}$ to $\cat{ZH}$. Hence, we have the following (non commutative) diagram:
$$\tikzfig{sop-interp-ncd}$$
We could hence define the interpretation from $\cat{SOP}$ to $\cat{ZH}$ as $\mathcal N(\interp{.})$. This would preserve the semantics, as $\mathcal N$ does. However, this would yield in general a diagram of exponential size in the size of the $\cat{SOP}$ morphism. Instead, we define in the following an interpretation of $\cat{SOP}$ morphisms as $\cat{ZH}$-diagrams of roughly the same size.

We define $\left[.\right]^{\operatorname{ZH}}:\cat{SOP}\to \cat{ZH}$ on arbitrary $\cat{SOP}$ morphisms as:
$$\left[s \sum_{\vec y} e^{2i\pi P}\ketbra{O_1,\ldots,O_m}{I_1,\ldots,I_n}\right]^{\operatorname{ZH}}:=\tikzfig{ZH-NF}$$
where the row of Z-spiders represents the variables $y_1,\ldots,y_k$.

The inputs of $O_i$ are linked to $y_1,\ldots,y_k$. The nodes $O_i$ can be inductively defined as:
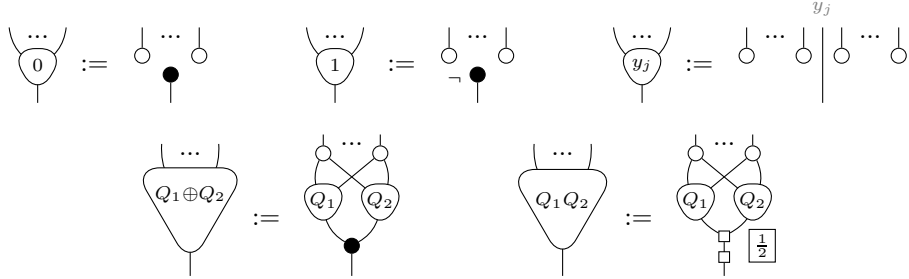
\begin{center}
\tikzfig{Q_i-0}\hspace*{3em}\tikzfig{Q_i-1}\hspace*{3em}\tikzfig{Q_i-y}\\[1em]
\tikzfig{Q_i-xor}\hspace*{3em}\tikzfig{Q_i-and}
\end{center}

The nodes $I_i$ are defined similarly, but upside-down.

The node $P$ can be inductively defined as:\\
\tikzfig{P-plus}\hfill\tikzfig{P-mon}\hfill~

The obtained diagram may be simplified using the simple ZH-rules:\\
\tikzfig{Z-spider-rule}\hfill\tikzfig{H-spider-rule}\hfill\tikzfig{ZH-scalar-mult}\hfill~\\\tikzfig{X-spider-rule-1}\hfill\tikzfig{X-spider-rule-2}\hfill\tikzfig{X-spider-rule-3}\hfill~

For instance, the polynomial $1\oplus x_1 \oplus x_1y_1y_2$ in a diagram that contains variables $x_1,x_2,y_1,y_2$ will be represented after simplification as: \tikzfig{polynomial-example}

\begin{example}
The $\cat{SOP}$ morphism:
$$\frac{1}{2\sqrt{2}}\sum\limits_{\vec y}e^{2i\pi\left(\frac14y_0+\frac{1}{2}y_4y_0+\frac{1}{8}y_5y_0y_1+\frac{3}{4}y_1y_2y_3+\frac{1}{2}y_0y_3\right)}\ketbra{0,1{\oplus} y_0{\oplus} y_4y_2,y_5}{y_4,y_5{\oplus} y_2{\oplus} y_3}$$
is mapped to \tikzfig{SOP-to-ZH-example}
\end{example}

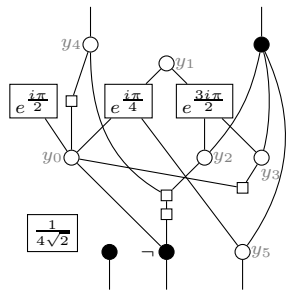
\begin{proposition}
\label{prop:double-interp-is-identity}
$\left[[.]^{\operatorname{ZH}}\right]^{\operatorname{sop}} \underset{\operatorname{Clif}}\sim (.)$
\end{proposition}

\begin{proof}
In appendix at page \pageref{prf:double-interp-is-identity}.
\end{proof}

\begin{corollary}
\label{cor:ZH-functor-preserves-semantics}
$\interp{[.]^{\operatorname{ZH}}} = \interp{.}$. In other words, the following diagram commutes:
$$\tikzfig{ZH-interp-cd}$$
\end{corollary}

\begin{proof}
Since $\underset{\operatorname{Clif}}\sim$ preserves the semantics, we have $\interp{.} = \interp{\left[[.]^{\operatorname{ZH}}\right]^{\operatorname{sop}}} = \interp{[.]^{\operatorname{ZH}}}$ by Propositions \ref{prop:double-interp-is-identity} and \ref{prop:sop-preserves-semantics}.
\end{proof}

\section{$\cat{SOP}$ for Clifford}
\label{sec:sop-clifford}

The \emph{Clifford} fragment of Quantum Mechanics is the one that represents $\cat{Stab}$. We would like to have a characterisation of this fragment for $\cat{SOP}$. Thankfully, this fragment is well known in $\cat{ZH}$. It can hence be inferred in $\cat{SOP}$ thanks to $[.]^{\operatorname{sop}}$.

\subsection{The Subcategories of $\cat{ZH}$ and $\cat{SOP}$ for Clifford}

\begin{definition}
$\cat{ZH}_{\operatorname{Clif}}$ is the $\dagger$-compact subPROP of $\cat{ZH}$ with the same objects, and generated by:\\
\tikzfig{Z-spider}, \tikzfig{H-2}, \tikzfig{H-phase} ($\alpha\in\{0,\frac\pi2,\pi,-\frac\pi2\}$), \tikzfig{H-scalar-1_sqrt2}, \tikzfig{H-scalar-omega}.
\end{definition}
Defining $\tikzfig{H-scalar-1_2}:=\tikzfig{H-scalar-1_sqrt2}~\tikzfig{H-scalar-1_sqrt2}$, we can still define the black spiders in this fragment.

\begin{proposition}
\label{prop:zh-clif-onto}
$\interp{.}:\cat{ZH}_{\operatorname{Clif}} \to \cat{Stab}$, the standard interpretation of $\cat{ZH}$-diagrams restricted to the Clifford fragment in $\cat{Stab}$ is onto.
\end{proposition}

\begin{proof}
In appendix at page \pageref{prf:zh-clif-onto}.
\end{proof}

We hence propose a restriction of $\cat{SOP}$ for the Clifford fragment, and show afterwards that it does indeed capture exactly $\cat{Stab}$.

\begin{definition}
$\cat{SOP}_{\operatorname{Clif}}$ is the subPROP of $\cat{SOP}$ with the same objects, and whose morphisms are of the form:
\[\frac{1}{\sqrt{2}^p}\sum e^{2i\pi\left(\frac{1}{8}P^{(0)} + \frac{1}{4}P^{(1)} + \frac{1}{2}P^{(2)}\right)}\ketbra{\vec O}{\vec I}\]
where $P^{(i)}$ is a polynomial with integer coefficients of degree at most $i$ (hence $P^{(0)}$ is in fact merely an integer); and where all the $O_i$ and $I_i$ are linear.
\end{definition}

\begin{proposition}
\label{prop:sop-clif-onto}
$\interp{.}:\cat{SOP}_{\operatorname{Clif}}\to\cat{Stab}$, the restriction of the standard interpretation to $\cat{SOP}_{\operatorname{Clif}}$ is onto $\cat{Stab}$.
\end{proposition}

\begin{proof}
In appendix at page \pageref{prf:sop-clif-onto}.
\end{proof}

Hence, $\cat{SOP}_{\operatorname{Clif}}$ does capture the Clifford fragment of quantum mechanics.

\subsection{A Complete Rewrite System for Clifford}

In \cite{SOP}, where the rewrite rules are introduced, the author gives a notion of completeness for Clifford \emph{unitaries}, that we will refer to in the following as ``weak completeness'':

\begin{proposition}[Weak Completeness for Clifford Unitaries]
Given two terms $t_1$, $t_2$ of $\cat{SOP}_{\operatorname{Clif}}$ such that $\interp{t_i}\circ\interp{t_i}^{\dagger}=id=\interp{t_i}^{\dagger}\circ\interp{t_i}$, we have:
$$t_1\circ t_2^\dagger \overset{*}{\rewrite{\operatorname{Clif}}} id ~~\iff~~ \interp{t_1} = \interp{t_2}$$
\end{proposition}

In practice, this is sufficient for deciding the equivalence of two Clifford quantum circuits, as they are represented as unitary morphisms of $\cat{SOP}_{\operatorname{Clif}}$. However, in our case, where we deal with more than unitaries, we cannot use this trick. Instead, we aim at a result like ``$t_1\overset{*}\longrightarrow t\overset{*}\longleftarrow t_2 \iff \interp{t_1} = \interp{t_2}$''. In other words, we want a rewrite system that will transform any term of $\cat{SOP}_{\operatorname{Clif}}$ into a unique normal form.

The rewrite system $\rewrite{\operatorname{Clif}}$ is not enough for this:
\begin{lemma}
There exist $t_1$ and $t_2$ two morphisms of $\cat{SOP}_{\operatorname{Clif}}$ such that:
\begin{itemize}
\item $\interp{t_1}=\interp{t_2}$
\item there is no $t_i'$ such that $t_i\rewrite{\operatorname{Clif}}t_i'$
\item $t_1\neq t_2$
\end{itemize}
\end{lemma}

\begin{proof}
An example of such behaviour can be obtained with:
$$t_1 := \sum e^{2i\pi\left(\frac{y_1y_2}2+\frac{y_2y}2\right)}\ketbra y{y_1,y_2}\qquad\qquad
t_2 := \sum e^{2i\pi\frac{y_2y}2}\ketbra{y_1{\oplus}y}{y_1,y_2}$$
\end{proof}

To palliate this problem, we propose to add three rewrite rules to the previously presented ones. These new rewrite rules are shown in Figure \ref{fig:rewrite-rules-2}.

\begin{figure}[!htb]
\[\sum_{\vec y} e^{2i\pi\left(P\right)}|O_1,\cdots,\overset{O_i}{\overbrace{y_0\oplus O_i'}},\cdots,O_m\rangle\!\!\bra{\vec I}
\rewrite{y_0\notin \Var(O_1,\ldots,O_{i-1},O_i')\\ O_i'\neq 0} \sum_{\vec y} e^{2i\pi\left(P[y_0\leftarrow \widehat{O_i}]\right)}\left(\ketbra{\vec O}{\vec I}\right)[y_0\leftarrow O_i]\tag{ket}\]
\[\sum_{\vec y} e^{2i\pi\left(P\right)}\ket{\vec O}\!\!\langle I_1,\cdots,\overset{I_i}{\overbrace{y_0\oplus I_i'}},\cdots,I_m|
\rewrite{y_0\notin \Var(\vec O,I_1,\ldots,I_{i-1},I_i')\\ I_i'\neq 0} \sum_{\vec y} e^{2i\pi\left(P[y_0\leftarrow \widehat{I_i}]\right)}\left(\ketbra{\vec O}{\vec I}\right)[y_0\leftarrow I_i]\tag{bra}\]

\[s\sum_{\vec y} e^{2i\pi\left(\frac{y_0}{2} + R\right)}\ketbra{\vec O}{\vec I}
\rewrite{R\neq 0 \text{ or } \vec O\vec I\neq \vec 0\\y_0\notin\Var(R,\vec O,\vec I)} \sum_{y_0} e^{2i\pi\left(\frac{y_0}{2}\right)}\ketbra{0,\cdots,0}{0,\cdots,0}\tag{Z}\]

\caption[]{Additional rewrite rules. Together with those of $\rewrite{\operatorname{Clif}}$, they constitute the rewrite system $\rewrite{\operatorname{Clif+}}$.}
\label{fig:rewrite-rules-2}
\end{figure}

The last rule (Z) describes what happens for a term that represents the linear map 0. Rule (bra) is simply the continuation of (ket). They explain how to operate suitable changes of variables.

\begin{proposition}
\label{prop:clif+-terminates}
The rewrite system $\rewrite{\operatorname{Clif+}}$ terminates.
\end{proposition}

\begin{proof}
In appendix at page \pageref{prf:clif+-terminates}.
\end{proof}

Not only does this rewrite system terminate, it is confluent in $\cat{SOP}_{\operatorname{Clif}}$ and the induced equivalence relation $\underset{\operatorname{Clif+}}\sim$ is complete for Clifford. We prove this by showing that any morphism of $\cat{SOP}_{\operatorname{Clif}}$ reduces to a normal form that is unique.

\begin{lemma}
\label{lem:NF}
Any morphism of $\cat{SOP}_{\operatorname{Clif}}$ reduces by $\rewrite{\operatorname{Clif+}}$ to a morphism of the form
$$\frac{1}{\sqrt{2}^{p}}\sum e^{2i\pi P}\ketbra{\vec O}{\vec I}$$
where:
\begin{itemize}
\item $\Var(P)\subseteq \Var(\vec O,\vec I)$ or $P=\frac{y_0}2$ where $y_0\notin \Var(\vec O,\vec I)$
\item $O_i = \begin{cases} y_k\\\quad\text{or}\\c\oplus\bigoplus\limits_{y\in \Var(O_1,...,O_{i-1})} c_yy \qquad\text{ where }c,c_y\in\{0,1\}\end{cases}$
\item $I_i = \begin{cases} y_k\\\quad\text{or}\\c\oplus\bigoplus\limits_{y\in \Var(\vec O, I_1,...,I_{i-1})} c_yy  \qquad\text{ where }c,c_y\in\{0,1\}\end{cases}$
\end{itemize}
\end{lemma}

\begin{proof}
In appendix at page \pageref{prf:NF}.
\end{proof}

To start with, we deal with the case where the term represents the null map.

\begin{proposition}
\label{prop:zero-morphism}
Let $t$ be a morphism of $\cat{SOP}_{\operatorname{Clif}}$ such that $\interp{t}=0$. Then:
$$t\overset\ast{\rewrite{\operatorname{Clif+}}} \sum\limits_{y_0} e^{2i\pi \frac{y_0}2}\ketbra{0,...,0}{0,...,0}$$
\end{proposition}

\begin{proof}
In appendix at page \pageref{prf:zero-morphism}.
\end{proof}

\begin{corollary}
\label{cor:non-zero-morphism}
If a morphism $t=\frac{1}{\sqrt{2}^{p}}\sum e^{2i\pi P}\ketbra{\vec O}{\vec I}$ of $\cat{SOP}_{\operatorname{Clif}}$ is irreducible such that $\Var(P)\subseteq \Var(\vec O,\vec I)$, then $\interp{t}\neq 0$.
\end{corollary}

Before moving on to the completeness by normal forms theorem, we need a result for the uniqueness of the phase polynomial:
\begin{lemma}
\label{lem:polynomial-uniqueness}
Let $P_1$ and $P_2$ be two polynomials of $\mathbb R[X_1,...,X_k]/(1,X^2-X)$, such that:
$$\forall \vec x\in\{0,1\}^k,~P_1(\vec x)=P_2(\vec x)$$
Then, $P_1 = P_2$.
\end{lemma}

\begin{proof}
In appendix at page \pageref{prf:polynomial-uniqueness}.
\end{proof}

\begin{theorem}
\label{thm:clif-completeness}
Let $t_1$, and $t_2$ be two morphisms of $\cat{SOP}_{\operatorname{Clif}}$ such that $\interp{t_1}=\interp{t_2}$. Then, there exists $t$ in $\cat{SOP}_{\operatorname{Clif}}$ such that $t_1\overset*{\rewrite{\operatorname{Clif+}}}t\overset*{\underset{\operatorname{Clif+}}\longleftarrow}t_2$, up to $\alpha$-conversion.
\end{theorem}

\begin{proof}
In appendix at page \pageref{prf:clif-completeness}.
\end{proof}

\begin{corollary}
The equality of morphisms in $\cat{SOP}_{\operatorname{Clif}}/\underset{\operatorname{Clif+}}\sim$ is decidable in time polynomial in the size of the phase polynomial and in the combined size of the ket/bra polynomials.
\end{corollary}

Although the set of rules is confluent in $\cat{SOP}_{\operatorname{Clif}}$, it is not in $\cat{SOP}$:

\begin{lemma}[Non-confluence]
The rewrite systems $\rewrite{\operatorname{Clif}}$ and $\rewrite{\operatorname{Clif+}}$ are not confluent in $\cat{SOP}$.
\end{lemma}

\begin{proof}
Consider the morphism $\sum e^{2i\pi\left(\frac{y_0}4+\frac{y_0y_1y_2}2+\frac{y_1y_3}2\right)}\ket{y_3}$:
\begin{align*}
&\sum e^{2i\pi\left(\frac{y_0}4+\frac{y_0y_1y_2}2+\frac{y_1y_3}2\right)}\ket{y_3}
\rewrite{\text{(HH)}}2\sum e^{2i\pi\left(\frac{y_0}4\right)}\ket{y_0y_2}
\end{align*}
However
\begin{align*}
\sum e^{2i\pi\left(\frac{y_0}4+\frac{y_0y_1y_2}2+\frac{y_1y_3}2\right)}\ket{y_3}
&\rewrite{(\omega)} \sqrt2\sum e^{2i\pi\left(\frac18-\frac{y_1y_2}4+\frac{y_1y_3}2\right)}\ket{y_3}
\end{align*}
The two resulting morphisms are not reducible, be it with $\rewrite{\operatorname{Clif}}$ or with $\rewrite{\operatorname{Clif+}}$.
\end{proof}

\subsection{Pivoting and Local Complementation}

We show here how, in the Clifford case, the rule (HH) corresponds to the operation of \emph{pivoting} \cite{pivoting}, and the rule ($\omega$) to that of \emph{local complementation} \cite{pi_2-complete,local-complementation}. To do so, we realise that graph states are easily representable in $\cat{SOP}$, for instance by interpreting the ZH-version of the graph state as a $\cat{SOP}$ morphism.

Let $G=(V,E)$ be a graph, with vertices $V$ and edges $E\subseteq V\times V$. The associated $\cat{SOP}$ morphism is:
$$ \sum_{\vec y\in V}e^{2i\pi\left(\sum\limits_{(y_i,y_j)\in E}\frac{y_iy_j}2\right)} \ket{\vec y}$$

\subsubsection*{Pivoting}
The operation of pivoting can be used to simplify a diagram of $\cat{ZH}_{\operatorname{Clif}}$ (or equivalently a Clifford diagram of the ZX-Calculus, as described in \cite{pivoting}). Informally, pivoting can be applied on any neighbouring pair of white nodes (where at least one of them is \emph{internal} i.e.~not linked to an input/output, for it to actually simplify the diagram). In the process, we complement the exclusive neighbours of both nodes with the other neighbours. Moreover, the common neighbours get an additional phase of $\pi$.

Let us see how it translates in $\cat{SOP}$. Let $t=s\sum e^{2i\pi \left(\frac{y_0y_i}2+\frac{y_0}2\widehat{Q_0}+\frac{y_i}2\widehat{Q_i}+\frac{y_0+y_i}2\widehat{Q_{0i}}+R\right)} \ketbra{\vec O}{\vec I}$ be a Clifford term, where the phase polynomial is already factorised by $y_0$ and $y_i$, the pair of variables/white dots on which to apply the pivoting. We consider that $y_0$ is internal $y_0\notin\Var(\vec O,\vec I)$. The fact that $y_0$ and $y_i$ are neighbours is captured by the term $\frac{y_0y_i}2$ in the phase polynomial. We can distinguish the exclusive neighbours of $y_0$ (resp.~$y_i$) by $Q_0$ (resp.~$Q_i$), and their common neighbours by $Q_{0i}$. 

The rule (HH) can be applied, with the substitution $[y_i\leftarrow Q_0\oplus Q_{0i}]$. The result is
$$t'=2s\sum e^{2i\pi \left(\frac12\widehat{Q_0Q_i}+\frac12\widehat{Q_{0i}Q_i}+\frac12\widehat{Q_0Q_{0i}}+\frac12\widehat{Q_{0i}}+R\right)} \left(\ketbra{\vec O}{\vec I}\right)[y_i\leftarrow Q_0\oplus Q_{0i}]$$
The term $\frac12\widehat{Q_0Q_i}$ creates the monomial $\frac{y_ky_\ell}2$ for all $y_k\in\Var(Q_0)$ and $y_\ell\in\Var(Q_i)$. If this monomial was already in $R$, it gets cancelled. This performs the complementation between the groups of variables in $Q_0$ and those in $Q_i$, and similarly for $\frac12\widehat{Q_{0i}Q_i}$ and $\frac12\widehat{Q_0Q_{0i}}$. On the other hand, the term $\frac12\widehat{Q_{0i}}$ creates a $\pi$ phase for all the common neighbours of $y_0$ and $y_i$.

\subsubsection*{Local Complementation}
The operation of local complementation is another operation that can be used to simplify the Clifford term at hand. Consider an internal white node in a Clifford diagram. If this node has a phase of $\pm\frac\pi2$, it can be removed. Doing so will add a phase of $\mp\frac\pi2$ to all the neighbours of the node, and at the same time, will perform a local complementation on them (all the nodes connected through an $H$ will get disconnected, and vice-versa). A global phase is also created.

A $\cat{SOP}$ morphism in this situation is of the form $t=s\sum e^{2i\pi \left(\frac{y_0}4+\frac{y_0}2(\sum x_i)+R\right)} \ketbra{\vec O}{\vec I}$ with $y_0$ an internal variable and $x_i$ its neighbours. The rule ($\omega$) can hence be applied, and the resulted term is:
$$t' = \sqrt2s\sum e^{2i\pi \left(\frac18-\frac14(\sum x_i)+\frac12(\sum\limits_{i\neq j} x_ix_j)+R\right)} \ketbra{\vec O}{\vec I}$$
as $-\frac14\widehat{\bigoplus x_i} = -\frac14(\sum x_i)+\frac12(\sum\limits_{i\neq j} x_ix_j) \bmod 1$.

The constant $\frac18$ corresponds to the global phase, the term $-\frac14(\sum x_i)$ represents an additional $-\frac\pi2$ phase to all the neighbours of $y_0$, and term $\frac12(\sum\limits_{i\neq j} x_ix_j)$ performs the local complementation on them.

In the case where $y_0$ holds a $-\frac\pi2$ phase, the term can also be simplified like this.

\section{$\cat{SOP}$ with Discards}
\label{sec:sop-discard}

\subsection{The Discard Construction on $\cat{SOP}$}

In \cite{CJPV19}, a construction is given to extend any $\dagger$-compact PROP for \emph{pure} quantum mechanics to another $\dagger$-compact PROP for quantum mechanics with environment. This new formalism can also be understood as the previous one, but where on top of it, one can discard the qubits. Because $\cat{SOP}$ fits the requirements, the construction can be applied to it.

First, we have to create the subcategory $\cat{SOP}_{\operatorname{iso}}$ of $\cat{SOP}$ that contains all its isometries. The objects of the new category are the same, and its morphisms are $\{f\in\cat{SOP}~|~\interp{f^\dagger\circ f}=id\}$.

These are important, as the isometries are exactly the pure quantum operators that can be discarded. The next step in the construction does just that. We perform the affine completion of $\cat{SOP}_{\operatorname{iso}}$, that is, for every object $n$, we add a new morphism $!_n:n\to0$, and we impose that $!\circ f = !$ for any $f$ in the new category, that we denote $\cat{SOP}_{\operatorname{iso}}^!$. We also need to impose that $!_n\otimes !_m = !_{n+m}$ and $!_0 = id_0$.

Finally, the category $\cat{SOP}^{\sground}$ is obtained as the pushout: \tikzfig{pushout} where the arrows are the inclusion functors.

We write morphisms in the new category in the form:
$$s\sum_{\vec y\in V^k} e^{2i\pi P(\vec y)}\ket{\vec O(\vec y)}!\vec D(\vec y)\bra{\vec I(\vec y)}$$
where the additional $D$ is a set of multivariate polynomials of $\mathbb F_2$. The fact that it is a set, and not a list, already captures some rules on the discard: first permuting qubits and then discarding them is equivalent to discarding them right away. Similarly, copying data and discarding the copies is equivalent to discarding the data right away.

Pure morphisms are those such that $\vec D = \{\}$. In those, no qubits are discarded. We hence easily induce usual morphisms such as $H$ and $\textit{CZ}$ in the new formalism.

The new morphisms $!_n$ are given by:
$$!_n := \sum_{\vec y\in V^n} \ket{}!\{y_1,\ldots,y_n\}\bra{y_1,\ldots,y_n}$$

In the new formalism, the compositions are obtained by:
$$f\circ g := \frac{s_fs_g}{2^{\abs{\vec I_f}}}\sum\limits_{\substack{\vec y_f,\vec y_g\\\vec y\in V^{\abs{\vec I_f}}}} e^{2i\pi \left(P_g+P_f+\frac{\vec O_g\cdot \vec y+\vec I_f\cdot \vec y}{2}\right)}\ket{\vec O_f}!D_f\cup D_g\bra{\vec I_g}$$
$$f\otimes g := s_fs_g\sum\limits_{\substack{\vec y_f,\vec y_g}} e^{2i\pi (P_g+P_f)}\ket{\vec O_f\vec O_g}!D_f\cup D_g\bra{\vec I_f\vec I_g}$$

It might be useful to be able to give an interpretation to the morphisms of the new formalism. To do so, we use the $\operatorname{CPM}$ construction \cite{Selinger-CPM} to map morphisms of $\cat{SOP}^{\sground}$ to morphisms of $\cat{SOP}$.

\begin{definition}
The functor $\operatorname{CPM}:\cat{SOP}^{\sground}/\underset{\operatorname{Clif}\raisebox{0.5ex}\sground}\sim\to\cat{SOP}/\underset{\operatorname{Clif+}}\sim$ is defined as:
\begin{align*}
s\sum_{\vec y} e^{2i\pi P}\ket{\vec O}!\vec D&\bra{\vec I} 
~\mapsto~\\
&\frac{s^2}{2^{\abs{\vec D}}}\sum_{\vec y_1,\vec y_2,\vec y} e^{2i\pi\left(P(\vec y_1)-P(\vec y_2)+ \frac{\vec D(\vec y_1)\cdot\vec y+\vec D(\vec y_2)\cdot\vec y}2 \right)}\ketbra{\vec O(\vec y_1),\vec O(\vec y_2)}{\vec I(\vec y_1),\vec I(\vec y_2)}
\end{align*}
\end{definition}


We can now define a standard interpretation of $\cat{SOP}^{\sground}$-morphisms as:
\begin{definition}
The standard interpretation $\interp{.}$ of $\cat{SOP}^{\sground}$ is defined as $\interp{.}:=\interp{\operatorname{CPM}(.)}$.
\end{definition}

Again, it is easy to transform any morphism of $\cat{SOP}^{\sground}$ in $\cat{ZH}^{\sground}$ and vice-versa: 
$$\left[s\sum_{\vec y\in V^k} e^{2i\pi P(\vec y)}\ket{\vec O(\vec y)}!\vec D(\vec y)\bra{\vec I(\vec y)}\right]^{\operatorname{ZH}} := \tikzfig{ZH-NF-ground}$$
and $\left[\ground\right]^{\operatorname{sop}} = !_1$.

\subsection{$\cat{SOP}$ with Discards for Clifford}

The discard construction can be applied to the subcategory $\cat{SOP}_{\operatorname{Clif}}$. We end up with a new category $\cat{SOP}_{\operatorname{Clif}}^{\sground}$, such that the following diagram, whose arrows are inclusions, commutes: $$\tikzfig{sop-discard-subcat-cd}$$

Following the characterisation of $\cat{SOP}_{\operatorname{Clif}}$ morphisms, we can determine that all the morphisms of $\cat{SOP}_{\operatorname{Clif}}^{\sground}$ are of the form:
$$ \frac1{\sqrt2^p}\sum e^{2i\pi\left(\frac18 P^{(0)}+\frac14 P^{(1)}+\frac12 P^{(2)}\right)}\ket{\vec O}!\vec D\bra{\vec I}$$
where $p\in\mathbb Z$, where $P^{(i)}$ is a polynomial with integer coefficients and of degree at most $i$, and where the polynomials of $\vec O,\vec D$ and $\vec I$ are linear.

The rewrite system presented previously can obviously be adapted to the new formalism (when there is a substitution, it has to be applied in $!\vec D$ as well). On top of that, the condition that makes $\cat{SOP}_{\operatorname{iso}}^!$ terminal can be translated as the meta rule:
\begin{align*}
\interp{\frac{s_f\overline{s_f}}{2^{\abs{O_f}}}\sum e^{2i\pi \left(P_f(\vec y)-P_f(\vec y')+\frac{\vec O_f(\vec y)\cdot \vec y''+ \vec O_f(\vec y')\cdot \vec y''}2\right)}\ketbra{\vec I(\vec y')}{\vec I(\vec y)}}=id\\
\implies \frac{s_fs_B}{2^{\abs{\vec O_B}}} \sum e^{2i\pi\left(P_B+P_f+\frac{\vec I_f\cdot \vec y+ \vec O_B\cdot \vec y}2\right)}\ket{\vec O}!\{\vec O_f,\ldots\}\bra{\vec I_B} = s_B \sum e^{2i\pi P_B}\ket{\vec O}!\{\vec O_B,\ldots\}\bra{\vec I_B}
\end{align*}
As you can see, this rule is not easy to apply. Thankfully, the last part of \cite{CJPV19} is devoted to showing that the big meta rule can sometimes be replaced by a few small ones. The idea is that, in some cases (in particular in the Clifford fragment), all the isometries can be generated from a finite set of generators. In particular, it is enough to impose that:
\begin{itemize}
\item $e^{i\alpha}=1$
\item $!\circ\ket0 = 1$
\item $!\circ H = !$\hfill (test)
\item $!\circ S = !$
\item $!_2\circ\textit{CZ}=!_2$
\end{itemize}
We give in Figure~\ref{fig:rewrite-rules-discard} the updated set of rewrite rules.

\begin{figure}[!htb]
\[\sum_{\vec y} e^{2i\pi P}\ket{\vec O}!\vec D\bra{\vec I}
\rewrite{y_0\notin\Var(P,\vec D,\vec O,\vec I)} 2\sum_{\vec y\setminus{\{y_0\}}} e^{2i\pi P}\ket{\vec O}!\vec D\bra{\vec I}\tag{Elim}\]

\[\sum_{\vec y} e^{2i\pi\left(\frac{y_0}{2} (y_0' + \widehat{Q}) + R\right)}\ket{\vec O}!\vec D\bra{\vec I}
\rewrite{y_0\notin\Var(R,Q,\vec D,\vec O,\vec I)\\y_0'\notin \Var(Q)} 2\!\!\sum_{\vec y\setminus{\{y_0,y_0'\}}}\!\! e^{2i\pi\left(R\left[y_0'\leftarrow \widehat{Q}\right]\right)}\left(\ket{\vec O}!\vec D\bra{\vec I}\right)\left[y_0'\leftarrow Q\right]\tag{HH}\]

\[\sum_{\vec y} e^{2i\pi\left(\frac{y_0}{4} + \frac{y_0}{2}\widehat{Q} + R\right)}\ket{\vec O}!\vec D\bra{\vec I}
\rewrite{y_0\notin\Var(Q,R,\vec D,\vec O,\vec I)} \sqrt{2}\sum_{\vec y\setminus{\{y_0\}}} e^{2i\pi\left(\frac{1}{8}-\frac{1}{4}\widehat{Q} + R\right)}\ket{\vec O}!\vec D\bra{\vec I}\tag{$\omega$}\]

\[\sum e^{2i\pi \left(P+\alpha \widehat{D_1\ldots D_p}\right)}\ket{\vec O}!\{D_1,\ldots,D_p,\ldots\}\bra{\vec I}
\rewrite{} \sum e^{2i\pi P}\ket{\vec O}!\{D_1,\ldots,D_p,\ldots\}\bra{\vec I}\tag{Z\sground}\]

\[\sum_{\vec y} e^{2i\pi \left(P+ \frac {y_0}2\widehat{Q}\right)}\ket{\vec O}!\vec D{\cup}\{y_0\} \bra{\vec I}
\rewrite{y_0\notin\Var(P,\vec O,\vec I, \vec D)} \sqrt2\sum_{\vec y\setminus\{y_0\}} e^{2i\pi P}\ket{\vec O}!\vec D{\cup}\{Q\}\bra{\vec I}\tag{H\sground}\]

\[\sum e^{2i\pi P}\ket{\vec O}!\{D_1,\ldots,D_p,D_1... D_p{\oplus}D_{p+1},\ldots\}\bra{\vec I}
\rewrite{}\sum e^{2i\pi P}\ket{\vec O}!\{D_1,\ldots,D_p,D_{p+1},\ldots\}\bra{\vec I}\tag{$\oplus$\sground}\]

\[\sum e^{2i\pi P}\ket{\vec O}!\{c,\ldots\}\bra{\vec I}
\rewrite{c\in\{0,1\}}\sum e^{2i\pi P}\ket{\vec O}!\{\ldots\}\bra{\vec I}\tag{Cst\sground}\]

\vspace*{-1em}
\begin{align*}
\sum_{\vec y} e^{2i\pi\left(P\right)}&\ket{\vec O}!\vec D\cup\{\overset{D_i}{\overbrace{y_0\oplus D_i'}}\}\bra{\vec I}\qquad\qquad\qquad\qquad
\end{align*}
\vspace*{-2.5em}
\begin{align*}
&\begin{array}{cl}
\rotatebox[origin=c]{-90}{$\rewrite{}$} & 
\begin{array}{l}
\scriptstyle{y_0\notin\Var(D_i')}\\
\scriptstyle{\abs{\{D_k{\in} \vec D~|\operatorname{mon}(D_k)\geq2\}}\geqslant|\{D_k{\in} \vec D[y_0{\leftarrow} D_i]~|\operatorname{mon}(D_k)\geq2\}|}
\end{array}
\end{array}\tag{disc}\\
\sum_{\vec y} e^{2i\pi\left(P[y_0\leftarrow \widehat{D_i}]\right)}&\left(\ket{\vec O}!\vec D\cup\{D_i\}\bra{\vec I}\right)[y_0\leftarrow D_i]
\end{align*}
\vspace*{-0.5em}
\[\sum_{\vec y} e^{2i\pi\left(P\right)}|\cdots,\overset{O_i}{\overbrace{y_0{\oplus} O_i'}}{\oplus} O_i'',\cdots\rangle !\vec D\bra{\vec I}
\hspace*{-4em}\rewrite{O_i'\neq 0\\y_0\notin \Var(O_1,\ldots,O_{i-1},O_i',O_i'')\\ y_0\notin\Var(\vec D)\text{ or }\{y_0, O_i'\}\subseteq \vec D{\cup}\{1\}}\hspace*{-1em}
\sum_{\vec y} e^{2i\pi\left(P[y_0\leftarrow \widehat{O_i}]\right)}\ket{\vec O[y_0{\leftarrow} O_i]}!\vec D\bra{\vec I[y_0{\leftarrow} O_i]}\tag{ket}\]
\vspace*{-0.5em}
\[\sum_{\vec y} e^{2i\pi\left(P\right)}\ket{\vec O}!\vec D\langle \cdots,\overset{I_i}{\overbrace{y_0{\oplus} I_i'}}{\oplus} I_i'',\cdots|
\hspace*{-2em}\rewrite{I_i'\neq 0\\y_0\notin \Var(\vec O,I_1,\ldots,I_{i-1},I_i',I_i'')\\ y_0\notin\Var(\vec D)\text{ or }\{y_0, I_i'\}\subseteq \vec D{\cup}\{1\}}\hspace*{-0.5em}
\sum_{\vec y} e^{2i\pi\left(P[y_0\leftarrow \widehat{I_i}]\right)}\ket{\vec O}!\vec D\bra{\vec I[y_0{\leftarrow} I_i]}\tag{bra}\]

\[s\sum_{\vec y} e^{2i\pi\left(\frac{y_0}{2} + R\right)}\ket{\vec O}!\vec D\bra{\vec I}
\rewrite{R\neq 0 \text{ or } \vec O\vec I\neq \vec 0\\y_0\notin\Var(R,\vec D,\vec O,\vec I)} \sum_{y_0} e^{2i\pi\left(\frac{y_0}{2}\right)}\ket{0,\cdots,0}!\{\}\bra{0,\cdots,0}\tag{Z}\]

\caption[]{Rewrite system $\rewrite{\operatorname{Clif}\raisebox{0.5ex}\sground}$ for $\cat{SOP}^{\sground}$.}
\label{fig:rewrite-rules-discard}
\end{figure}

Notice that we have made the choice to simplify the discarded polynomials before those in kets and bras. This is motivated by the example:
\begin{example}
Consider $t:= \ket{y_1,y_2,y_3}!\{y_1{\oplus}y_2,y_2{\oplus}y_3,y_1{\oplus}y_3\}$. If (ket) had priority over (disc), the term could not be reduced. Instead, we reduce $t$ as:
\begin{align*}
\ket{y_1,y_2,y_3}!\{y_1{\oplus}y_2,y_2{\oplus}y_3,y_1{\oplus}y_3\}\rewrite{\text{(disc)}}\ket{y_1{\oplus}y_2, y_2,y_3}!\{y_1,y_2{\oplus}y_3,y_1{\oplus}y_2{\oplus}y_3\}\\
\rewrite{\oplus\sground}\ket{y_1{\oplus}y_2, y_2,y_3}!\{y_1,y_2{\oplus}y_3\}
\rewrite{\text{(disc)}}\ket{y_1{\oplus}y_2, y_2,y_2{\oplus}y_3}!\{y_1,y_3\}\\
\rewrite{\text{ket}}\ket{y_2, y_1{\oplus}y_2,y_1{\oplus}y_2{\oplus}y_3}!\{y_1,y_3\}
\rewrite{\text{ket}}\ket{y_2, y_1{\oplus}y_2,y_2{\oplus}y_3}!\{y_1,y_3\}
\end{align*}
Hence, by giving priority to (disc) over (ket) and (bra), one can hope to reduce the number of discarded polynomials.
\end{example}

%

\begin{proposition}
\label{prop:clif!-terminates}
The rewrite system $\rewrite{\operatorname{Clif}\raisebox{0.5ex}\sground}$ terminates.
\end{proposition}

\begin{proof}
In appendix at page \pageref{prf:clif!-terminates}.
\end{proof}

\begin{lemma}
\label{lem:clif-disc-form}
Any non-null morphism of $\cat{SOP}_{\operatorname{Clif}}^{\sground}$ can be reduced to:
$$ \frac1{\sqrt2^p}\sum_{\vec y,\vec y_d} e^{2i\pi\left(\frac14 P^{(1)}(\vec y)+\frac12 P^{(2)}(\vec y,\vec y_d)\right)}\ket{\vec O(\vec y,\vec y_d)}!\{\vec y_d\}\bra{\vec I(\vec y,\vec y_d)}$$
where:
\begin{itemize}
\item polynomials of $\vec O$ and $\vec I$ are linear
\item the set of discarded polynomials is reduced to a set of variables $\{\vec y_d\}$
\item $P^{(1)}$ and $P^{(2)}$ have no constants
\item no monomial of $P^{(2)}$ uses only variables of $\vec y_d$
\item $\{\vec y_d\}\subseteq\Var(\vec O,\vec I)$ i.e.~discarded variables have to appear somewhere in the ket or bra
\item $\Var(P^{(1)},P^{(2)}) \subseteq \Var(\vec O,\vec I,\vec D)$ or $P = \frac{y_0}2$ with $y_0\notin\Var(\vec O,\vec I,\vec D)$.
\end{itemize}
\end{lemma}

\begin{proof}
In appendix at page \pageref{prf:clif-disc-form}.
\end{proof}

\begin{corollary}
Any morphism of $\cat{SOP}_{\operatorname{Clif}}^{\sground}$ eventually reduces to a morphism of the form given in Lemma \ref{lem:clif-disc-form}.
\end{corollary}

\begin{proof}
As the rewrite system terminates, and since every morphism of $\cat{SOP}_{\operatorname{Clif}}^{\sground}$ can be reduced into the form of Lemma \ref{lem:clif-disc-form}, the rewrite system terminates in a term of the form of Lemma \ref{lem:clif-disc-form}.
\end{proof}

\begin{lemma}
\label{lem:zero-morphisms-disc}
Any morphism $t$ of $\cat{SOP}_{\operatorname{Clif}}^{\sground}$ such that $\interp{t}=0$ reduces to:
$$\sum_{y_0} e^{2i\pi\left(\frac{y_0}{2}\right)}\ket{0,\cdots,0}!\{\}\bra{0,\cdots,0}$$
\end{lemma}

\begin{proof}
In appendix at page \pageref{prf:zero-morphisms-disc}.
\end{proof}

\begin{corollary}
\label{cor:non-zero-morphisms-disc}
If a morphism $t$ of $\cat{SOP}_{\operatorname{Clif}}^{\sground}$ is terminal with $\Var(P)\subseteq\Var(\vec O,\vec D,\vec I)$, then $\interp{t}\neq 0$.
\end{corollary}

\begin{theorem}[Completeness for Clifford]
\label{thm:disc-completeness}
Let $t_1$ and $t_2$ be two morphisms of $\cat{SOP}_{\operatorname{Clif}}^{\sground}$ such that $\interp{t_1} = \interp{t_2}$. If $t_1'$ and $t_2'$ are terminal such that $t_1\overset{*}{\rewrite{\operatorname{Clif}\raisebox{0.5ex}\sground}}t_1'$ and $t_2\overset{*}{\rewrite{\operatorname{Clif}\raisebox{0.5ex}\sground}}t_2'$, then $t_1' = t_2'$ up to $\alpha$-conversion.
\end{theorem}

To prove this theorem, we suggest to use the similar result in $\cat{SOP}_{\operatorname{Clif}}$, and transport it to our case. To do so, we need some additional constructions.

\begin{definition}
We define $\overline{\cat{SOP}_{\operatorname{Clif}}^{\sground}}$ as the set of morphisms of $\cat{SOP}_{\operatorname{Clif}}^{\sground}$ in the form given in Lemma \ref{lem:clif-disc-form}. We define the function $F$ on $\overline{\cat{SOP}_{\operatorname{Clif}}^{\sground}}$ such that, for any morphism
$$t= \frac1{\sqrt2^{p}}\sum_{\vec y,\vec y_d} e^{2i\pi P(\vec y,\vec y_d)}\ket{\vec O(\vec y,\vec y_d)}!\{\vec y_d\}\bra{\vec I(\vec y,\vec y_d)}$$
of $\overline{\cat{SOP}_{\operatorname{Clif}}^{\sground}}$:
$$F(t):= \frac1{\sqrt2^{2p}}\sum_{\vec y,\vec y',\vec y_d} e^{2i\pi\left(P(\vec y,\vec y_d)-P(\vec y',\vec y_d)\right)}\ketbra{\vec O(\vec y,\vec y_d),\vec O(\vec y',\vec y_d)}{\vec I(\vec y,\vec y_d),\vec I(\vec y',\vec y_d)}$$
\end{definition}

This new functor $F$ can be seen as a simplified CPM construction, in the case where the term is already simplified.

\begin{proposition}
\label{prop:F-is-CPM}
For any $t\in\overline{\cat{SOP}_{\operatorname{Clif}}^{\sground}}$, $F(t){\underset{\operatorname{Clif+}}\sim}\operatorname{CPM}(t)$. This implies $\interp{F(.)}=\interp{\operatorname{CPM}(.)}$.
\end{proposition}

\begin{proof}
\begin{align*}
&\operatorname{CPM}(t) = \frac{1}{2^{p+\abs{\vec y_{d_1}}}}\sum_{\vec y_1,\vec y_{d_1},\vec y_2,\vec y_{d_2},\vec y}\hspace*{-2em} e^{2i\pi\left(P(\vec y_1,\vec y_{d_1})-P(\vec y_2,\vec y_{d_2})+ \frac{\vec y_{d_1}\cdot\vec y+\vec y_{d_2}\cdot\vec y}2 \right)}\hspace*{-8em}\raisebox{-0.8em}{$\ketbra{\vec O(\vec y_1,\vec y_{d_1}),\vec O(\vec y_2,\vec y_{d_2})}{\vec I(\vec y_1,\vec y_{d_1}),\vec I(\vec y_2,\vec y_{d_2})}$}\\
&\overset\ast{\rewrite{\text{(HH)}}}\frac1{\sqrt2^{2p}}\sum_{\vec y_1,\vec y_2,\vec y_{d_1}} e^{2i\pi\left(P(\vec y_1,\vec y_{d_1})-P(\vec y_2,\vec y_{d_1})\right)}\ketbra{\vec O(\vec y_1,\vec y_{d_1}),\vec O(\vec y_2,\vec y_{d_1})}{\vec I(\vec y_1,\vec y_{d_1}),\vec I(\vec y_2,\vec y_{d_1})} = F(t)
\end{align*}
\end{proof}

\begin{definition}
We define the function $G$ on some morphisms of $\cat{SOP}_{\operatorname{Clif}}$.\\
Let $t=\frac1{\sqrt2^p}\sum_{\vec y}e^{2i\pi P}\ketbra{\vec O_1,\vec O_2}{\vec I_1,\vec I_2}$ such that:
\begin{itemize}
\item $p = 2p'$
\item $\abs{\vec O_1}=\abs{\vec O_2}$ and $\abs{\vec I_1}=\abs{\vec I_2}$
\item $\{\vec y_d\} := \{\vec y\} \setminus \Var(\vec O_1\oplus \vec O_2, \vec I_1\oplus \vec I_2)$
\item $\{\vec y_1\}:= \Var(\vec O_1,\vec I_1)\setminus \{\vec y_d\}$
\item $\{\vec y_2\}:=\left(\{\vec y\}\setminus\{\vec y_1\}\right)\setminus\{\vec y_d\}$
\item $\abs{\vec y_1}=\abs{\vec y_2}$
\item there exists a unique bijection $\delta:\{\vec y_2\}\to\{\vec y_1\}$ such that $(\vec O_1\oplus \vec O_2, \vec I_1\oplus \vec I_2)[\vec y_2\leftarrow \delta(\vec y_2)]=\vec 0$
\end{itemize}
then $G(t)$ is defined, and:
$$G(t):=\frac1{\sqrt2^{p'}}\sum_{\vec y_1,\vec y_d}e^{-2i\pi P[\vec y_1{\leftarrow} \vec 0][\vec y_2{\leftarrow}\delta(\vec y_2)]}\ket{\vec O_2[\vec y_1{\leftarrow} \vec 0][\vec y_2{\leftarrow}\delta(\vec y_2)]}!\{\vec y_d\}\bra{\vec I_2[\vec y_1{\leftarrow} \vec 0][\vec y_2{\leftarrow}\delta(\vec y_2)]}$$
\end{definition}

The function $G$ is designed to be an inverse of $F$ for some morphisms, while at the same being impervious to some rewrite rules.

\begin{proposition}
\label{prop:G-cd}
Let $t$ be terminal with $\rewrite{\operatorname{Clif}\raisebox{0.5ex}\sground}$. Then, the following diagram commutes up to $\alpha$-conversion:
$$\tikzfig{disc-cpm-cd}$$
for any $t'$ obtained by reducing $F(t)$.
\end{proposition}

\begin{proof}
In appendix at page \pageref{prf:G-cd}.
\end{proof}

\begin{proof}[Theorem \ref{thm:disc-completeness}]
Let $t_1$ and $t_2$ be two morphisms of $\cat{SOP}_{\operatorname{Clif}}^{\sground}$ such that $\interp{t_1}=\interp{t_2}$. Since $\rewrite{\operatorname{Clif}\raisebox{0.5ex}\sground}$ terminates by Proposition \ref{prop:G-cd}, both $t_1$ and $t_2$ reduce to respectively $t_1'$ and $t_2'$, two terminal morphisms of $\overline{\cat{SOP}_{\operatorname{Clif}}^{\sground}}$. By soundness, $\interp{t_1'}=\interp{t_2'}$, so, by Proposition \ref{prop:F-is-CPM}, $\interp{F(t_1')}=\interp{F(t_2')}$. By completeness of $\rewrite{\operatorname{Clif+}}$, we have $F(t_1')\overset*{\rewrite{\operatorname{Clif+}}}t'\overset*{\underset{\operatorname{Clif+}}\longleftarrow}F(t_2')$ up to $\alpha$-conversion. Finally, by Proposition \ref{prop:G-cd}, $t_1'=G(t')=t_2'$ up to $\alpha$-conversion:
$$\tikzfig{completeness-disc-cd}$$
\end{proof}

\begin{remark}
Interestingly, the previous proposition and theorem show that the simplification of a term of $\cat{SOP}_{\operatorname{Clif}}^{\sground}$ can be operated in the ``pure'' setting, and then $G$ can be used to retrieve the normal form. More precisely:
$$\tikzfig{ground-from-CPM-cd}$$
\end{remark}

\begin{corollary}
The equality of morphisms in $\cat{SOP}_{\operatorname{Clif}}^{\sground}/\underset{\operatorname{Clif}\raisebox{0.5ex}\sground}\sim$ is decidable in time polynomial in the size of the phase polynomial and in the combined size of the ket/bra/discarded polynomials.
\end{corollary}

\section*{Conclusion and Further Work}

We have shown that $\cat{SOP}$ could represent any morphism of $\cat{Qubit}$, and that it could be enriched using the discard construction to include measurements. We have shown a correspondence between this formalism and graphical languages such as the ZH-Calculus, and we have provided two rewrite strategies for simplifying terms. We have shown that these are complete in the Clifford case.

This framework can be used to simplify Z*-diagrams: one simply needs to translate the diagram as a $\cat{SOP}$-morphism, simplify it, then translate the result as a diagram in the target language.

By applying the discard construction, we have extended the domain of use of $\cat{SOP}$ to programs that contain measurements. For instance, schemes for error detection/correction can now be studied/verified/simplified in the framework.

One of the obvious further developments of the framework is to use the completeness of (fragments of) Z*-Calculi and their interpretation to generate rewrite strategies complete for fragments larger than Clifford. On can also transport constructions that are known in the Z*-Calculi to perform non trivial operations on $\cat{SOP}$ morphisms.

Another important development of the framework would be to more easily represent families of processes. The recent enrichment SZX \cite{SZX} could be of help for this topic.

Finally, it could be interesting to see how graph-theoretic notions like the gflow \cite{mbqc} translate to $\cat{SOP}$. This particular notion could for instance allow to extract a quantum circuit from an arbitrary (isometry) $\cat{SOP}$-morphism.

\section*{Acknowledgements}
The author acknowledges support from the project PIA-GDN/Quantex. The author would like to thank Simon Perdrix, Emmanuel Jeandel and Benoît Valiron for fruitful discussions.


\appendix
\section{Appendix}

\begin{proof}[Proposition \ref{prop:transpose-conjugate-dagger}]
\phantomsection
\label{prf:transpose-conjugate-dagger}~
\begin{itemize}
\item $\interp{(.)^t} = \interp{(\epsilon_m \otimes id_n)\circ(id_m\otimes .\otimes id_n)\circ (id_m\otimes \eta_m)} = (\epsilon_m \otimes id_n)\circ(id_m\otimes \interp{.}\otimes id_n)\circ (id_m\otimes \eta_m)$ \\ \hspace*{2.5em}$= \interp{.}^t$\medskip
\item $\interp{\overline{f}}= {s_f}\sum\limits_{\vec y,\vec x\in\{0,1\}} e^{-2i\pi P_f}\ketbra{\vec Q_f}{\vec x} = \overline{s_f\sum\limits_{\vec y,\vec x\in\{0,1\}} e^{2i\pi P_f}\ketbra{\vec Q_f}{\vec x}} = \overline{\interp{f}}$\medskip
\item $\interp{(.)^\dagger} = \interp{\overline{(.)}^t} = \overline{\interp{.}}^t = \interp{.}^\dagger$
\end{itemize}
\end{proof}

\begin{proof}[Theorem \ref{thm:dagger-compact-prop}]
\phantomsection
\label{prf:dagger-compact-prop}
As we already mentioned, $\cat{SOP}$ is a PROP. Quotienting it with the equivalence relation $\underset{\operatorname{Clif}}\sim$ does not change this property. We already saw that all the axioms for a compact structure are satisfied. It remains to show that $(.)^\dagger$ is involutive.

First, for any morphism $f\in\cat{SOP}$, we have:
$$\overline{\overline{f}}= \overline{\ket{\vec x}\mapsto{s_f}\sum e^{-2i\pi P_f}\ket{\vec Q_f}} = \ket{\vec x}\mapsto s_f\sum e^{2i\pi P_f}\ket{\vec Q_f} = f$$
so $\overline{\overline{(.)}} = (.)$.

Moreover:
\begin{align*}
\overline{f}^t &= \ket{\vec x}\mapsto \frac{s_f}{2^m}\sum e^{2i\pi \left(-P_f+\frac{\vec{\widehat{Q_f}}[\vec x_f \leftarrow \vec y]\cdot \vec y'+\vec x\cdot\vec y'}2\right)}\ket{\vec y} \\
&= \ket{\vec x}\mapsto \frac{s_f}{2^m}\sum e^{2i\pi \left(-P_f-\frac{\vec{\widehat{Q_f}}[\vec x_f \leftarrow \vec y]\cdot \vec y'+\vec x\cdot\vec y'}2\right)}\ket{\vec y} = \overline{f^t}
\end{align*}
Indeed, $\vec{\widehat{Q_f}}[\vec x_f \leftarrow \vec y]\cdot \vec y'+\vec x\cdot\vec y'$ is integer-valued, and $e^{i\pi n} = e^{-i\pi n}$ for any $n\in\mathbb Z$.

Finally: $(.)^{\dagger\dagger} = \overline{\overline{(.)}^t}^t = \overline{\overline{(.)}}^{tt} = (.)$
\end{proof}

\begin{proof}[Lemma \ref{lem:H-spider-decomp}]
\phantomsection
\label{prf:H-spider-decomp}
First, one of the (sound) rules of the ZH-Calculus tells us that:
\begin{align*}
\interp{\tikzfig{H-spider}} = \interp{\tikzfig{H-spider-decomp-aux}}
\end{align*}
Then, let $\rho e^{i\theta} := \frac{1-r}{1+r}$ with $\rho>0$, which is well defined and unique since $r\neq-1$ and $r\neq1$. Let also $\alpha:=2\arctan\frac\rho2$, $\beta:=\theta+\frac\pi2$ and $s:=\frac{1+r}{2(1+e^{i\alpha})}$. Then, one can check that:
\def\fig{H-spider-decomp-proof}
\begin{align*}
\interp{\input{./figures/\fig/\fig_00.tikz}}
=\interp{\input{./figures/\fig/\fig_01.tikz}}
=\interp{\input{./figures/\fig/\fig_02.tikz}}
\end{align*}
\end{proof}

\begin{proof}[Proposition \ref{prop:double-interp-is-identity}]
\phantomsection
\label{prf:double-interp-is-identity}
First, we can show inductively that $\left[\tikzfig{P}\right]^{\operatorname{sop}} \underset{\operatorname{Clif}}\sim \sum_{\vec y} e^{2i\pi P(\vec y)}\ketbra{}{\vec y}$. Indeed, we have:
\[\left[\tikzfig{P-mon-1}\right]^{\operatorname{sop}} = \left[\tikzfig{P-mon-2}\right]^{\operatorname{sop}} = \sum e^{2i\pi\alpha y_{i_1}\ldots y_{i_s}}\ketbra{}{\vec y}\]
and
\begin{align*}
\left[\tikzfig{P-plus-1}\right]^{\operatorname{sop}} &= \left[\tikzfig{P-plus-2}\right]^{\operatorname{sop}} = \left(\sum e^{2i\pi(P_1(\vec y_1)+P_2(\vec y_2))}\ketbra{}{\vec y_1,\vec y_2}\right)\circ\left(\sum \ketbra{\vec y,\vec y}{\vec y}\right)\\
&\underset{\operatorname{Clif}}\sim \sum e^{2i\pi(P_1(\vec y)+P_2(\vec y))}\ketbra{}{\vec y} = \sum e^{2i\pi(P_1+P_2)(\vec y)}\ketbra{}{\vec y}
\end{align*}

Similarly, we can prove that $\left[\tikzfig{O_i}\right]^{\operatorname{sop}} \underset{\operatorname{Clif}}\sim \sum_{\vec y}\ketbra{O_i(\vec y)}{\vec y}$. The base cases are straightforward, so we show the sum and product. Notice that:
$$\left[\tikzfig{xor}\right]^{\operatorname{sop}} = \frac14\sum e^{2i\pi\left(\frac{y_5y_1}2+\frac{y_6y_2}2+\frac{y_1y_3}2+\frac{y_2y_3}2+\frac{y_1y_4}2\right)}\ketbra{y_4}{y_5,y_6} \overset *{\rewrite{$Elim$}} \sum \ketbra{y_1\oplus y_2}{y_1,y_2}$$
$$\left[\tikzfig{and}\right]^{\operatorname{sop}} = \frac12\sum e^{2i\pi\left(\frac{y_3y_4y_1}2+\frac{y_1y_2}2\right)}\ketbra{y_2}{y_3,y_4} \rewrite{$Elim$} \sum \ketbra{y_1y_2}{y_1,y_2}$$
and
$$\left[\tikzfig{Q_1-Q_2}\right]^{\operatorname{sop}} \underset{\operatorname{Clif}}\sim \sum \ketbra{Q_1(\vec y),Q_2(\vec y)}{\vec y}$$
so we directly get:
$$\left[\tikzfig{Q_i-xor-1}\right]^{\operatorname{sop}} = \left[\tikzfig{Q_i-xor-2}\right]^{\operatorname{sop}} \underset{\operatorname{Clif}}\sim \sum \ketbra{(Q_1\oplus Q_2)(\vec y)}{\vec y}$$
and
$$\left[\tikzfig{Q_i-and-1}\right]^{\operatorname{sop}} = \left[\tikzfig{Q_i-and-2}\right]^{\operatorname{sop}} \underset{\operatorname{Clif}}\sim \sum \ketbra{(Q_1 Q_2)(\vec y)}{\vec y}$$
Finally:
\tikzfig{ZH-NF-proof-bis}
Hence:
$$\left[\tikzfig{ZH-NF}\right]^{\operatorname{sop}}\underset{\operatorname{Clif}}\sim \left[\tikzfig{ZH-NF-M-S}\right]^{\operatorname{sop}}\underset{\operatorname{Clif}}\sim \sum e^{2i\pi P(\vec y)}\ketbra{\vec O(\vec y)}{\vec I(\vec y)}$$
\end{proof}

\begin{proof}[Proposition \ref{prop:zh-clif-onto}]
\phantomsection
\label{prf:zh-clif-onto}
First, we shall show that $\interp{\cat{ZH}_{\operatorname{Clif}}} \subseteq \cat{Stab}$. To do so, it suffices to show that all the generators of $\cat{ZH}_{\operatorname{Clif}}$ are mapped to morphisms of $\cat{Stab}$:
\begin{align*}
\interp{~\tikzfig{H-scalar-1_sqrt2}~} &= \frac1{\sqrt2} = \epsilon\circ(\ket0\otimes (H\circ \ket0)) \in\cat{Stab}\\
2 &= \epsilon\circ\eta \in\cat{Stab}\\
\interp{\tikzfig{H-2}}&= H\times\frac2{\sqrt2}\in\cat{Stab}\\
\interp{\tikzfig{H-phase-0}} &= \ket0+\ket1 = \frac2{\sqrt2}H\ket0\in\cat{Stab}\\
\interp{\tikzfig{H-phase-k-pi_2}} &= \ket0 + i^k\ket1 = S^k(\ket0+\ket1)\in\cat{Stab}\\
\interp{\tikzfig{Z-1-2}} &= (id\otimes H)\circ \textit{CZ}\circ (id\otimes (H\circ\ket0))\in\cat{Stab}\\
\interp{\rotatebox[origin=c]{180}{\tikzfig{Z-1-2}}} &= (id\otimes \epsilon)\circ\left(\interp{\tikzfig{Z-1-2}}\otimes id\right)\in\cat{Stab}\\
\interp{\tikzfig{Z-1-0}} &= \epsilon\circ \interp{\tikzfig{Z-1-2}}\in\cat{Stab}\\
\interp{\rotatebox[origin=c]{180}{\tikzfig{Z-1-0}}} &= \interp{\rotatebox[origin=c]{180}{\tikzfig{Z-1-2}}}\circ\eta\in\cat{Stab}\\
\interp{~\tikzfig{H-scalar-omega}~} &= \interp{\tikzfig{H-scalar-omega-2}}\in\cat{Stab}
\end{align*}
and $\interp{\tikzfig{Z-spider}}$ can be obtained as a composition of $\interp{\tikzfig{Z-1-2}}$, $\interp{\rotatebox[origin=c]{180}{\tikzfig{Z-1-2}}}$, $\interp{\tikzfig{Z-1-0}}$ and $\interp{\rotatebox[origin=c]{180}{\tikzfig{Z-1-0}}}$.

Then, we can show that all the generators of $\cat{Stab}$ have a preimage by $\interp{.}_{\operatorname{Clif}}$ in $\cat{ZH}_{\operatorname{Clif}}$:
\begin{align*}
\interp{\tikzfig{ket0}}&=\ket0\\
\interp{~\tikzfig{H-2}~~\tikzfig{H-scalar-1_sqrt2}~}&= H\\
\interp{\tikzfig{CZ}}&= \textit{CZ}\\
\interp{\tikzfig{S-gate}}&= S
\end{align*}
\end{proof}

\begin{proof}[Proposition \ref{prop:sop-clif-onto}]
\phantomsection
\label{prf:sop-clif-onto}
First, we show that $\interp{\cat{SOP}_{\operatorname{Clif}}}\subseteq\cat{Stab}$. To do so, it can be seen that $[\cat{SOP}_{\operatorname{Clif}}]^{\operatorname{ZH}}\subseteq \cat{ZH}_{\operatorname{Clif}}$:
$\frac1{\sqrt2^p}$ is mapped to $\tikzfig{H-scalar-1_sqrt2}^p$, $\frac18P^{(0)}$ contributes for $\tikzfig{H-scalar-omega}^{P^{(0)}}$, $\frac14P^{(1)}$ contributes to \tikzfig{H-phase} linked to the associated variable (where $\alpha\in\{0,\frac\pi2,\pi,-\frac\pi2\}$), and $\frac12P^{(2)}$ contributes to $\tikzfig{H-2}$ linked to the associated pair of variables; and $O_i$ and $I_i$ being linear, they are mapped to black spiders (with or without $\neg$).\\
Hence, $\interp{\cat{SOP}_{\operatorname{Clif}}}=\interp{[\cat{SOP}_{\operatorname{Clif}}]^{\operatorname{ZH}}}\subseteq\interp{\cat{ZH}_{\operatorname{Clif}}}\subseteq \cat{Stab}$.

Next, it suffices to show that all the generators of $\cat{Stab}$ have a preimage by $\interp{.}$ in $\cat{SOP}_{\operatorname{Clif}}$:
\begin{align*}
\interp{\frac1{\sqrt2}\sum e^{2i\pi\frac{y_1y_2}2}\ketbra{y_2}{y_1}} &= H\\
\interp{\sum e^{2i\pi\frac y4}\ketbra yy} &= S\\
\interp{\sum e^{2i\pi\frac{y_1y_2}{2}}\ketbra{y_1,y_2}{y_1,y_2}} &= \textit{CZ}\\
\interp{\ketbra0{}} &= \ket0
\end{align*}
\end{proof}

\begin{proof}[Proposition \ref{prop:clif+-terminates}]
\phantomsection
\label{prf:clif+-terminates}
For a morphism $s\sum_{\vec y}\limits e^{2i\pi P}\ketbra{\vec O}{\vec I}$ of $\cat{SOP}$, consider the tuple:
$$\left(\vphantom{\rule{0.5pt}{1em}} \abs{\vec y}, ~\operatorname{mon}(O_1), ~\ldots, ~\operatorname{mon}(O_m), ~\operatorname{mon}(I_1), ~\ldots, ~\operatorname{mon}(I_n), \operatorname{mon}(P)\right)$$
where:
\begin{itemize}
\item $\abs{.}$ denotes the cardinality
\item $\operatorname{mon}(Q)$ counts the number of monomials in the expanded simplified polynomial $Q$
\end{itemize}

We can define an order on these tuples, as their lexicographic order. Notice that all the components of the tuple are natural integers. Hence, if we can show that every rewrite rule in $\rewrite{\operatorname{Clif+}}$ strictly reduces the tuple, then it means $\rewrite{\operatorname{Clif+}}$ terminates.

It is easy to check that the three rules of $\rewrite{\operatorname{Clif}}$ reduce the size of $\vec y$, hence reducing the tuple.

%

When the rule (ket) is applied on $O_i$, we necessarily have $\operatorname{mon}(O_i)\geq 2$. Indeed, $O_i = y_0\oplus O_i'$ where $O_i'\neq 0$. After application of the rule, this quantity is reduced to $1$. Moreover, neither $\abs{\vec y}$, $\operatorname{mon}(O_1)$, ..., nor $\operatorname{mon}(O_{i-1})$ is changed, as there is no creation or removal of variables, and $y_0$ does not appear in $O_1,\ldots,O_{i-1}$. The rule (bra) works exactly in the same fashion.

Finally, the rule (Z) reduces the morphism to one whose tuple is $(1,0,\ldots,0,0,\ldots,0,1)$, and only from a morphism with a larger associated tuple.
\end{proof}

\begin{proof}[Lemma \ref{lem:NF}]
\phantomsection
\label{prf:NF}
The rules (ket) and (bra) quite obviously enforce the form of $\vec O$ and $\vec I$. Then, suppose $y_0$ is an internal variable. Then either:
\begin{itemize}
\item the monomial $\frac14y_0$ appears in the phase polynomial, in which case the rule ($\omega$) can be applied
\item the monomial $\frac12y_0y_i$ appears in the phase polynomial, with some arbitrary $y_i$, in which case the rule (HH) can be applied
\item the monomial $\frac12y_0$ appears in the phase polynomial, as the only occurrence of $y_0$, in which case the rule (Z) can be applied
\end{itemize}
\end{proof}

\begin{proof}[Proposition \ref{prop:zero-morphism}]
\phantomsection
\label{prf:zero-morphism}
By reductio ad absurdum, suppose that $t$ reduces to $t'= \frac{1}{\sqrt{2}^{p}}\!\sum e^{2i\pi P}\!\ketbra{\vec O}{\vec I}$ different from $\sum\limits_{y_0} e^{2i\pi \frac{y_0}2}\ketbra{0,...,0}{0,...,0}$, but irreducible. By Lemma \ref{lem:NF}, this implies $\Var(P)\subseteq \Var(\vec O,\vec I)$. We show that we can build $\vec x_1,\vec x_2 \in\{0,1\}^{n+m}$ such that $\bra{\vec x_1}\interp{t'}\ket{\vec x_2}\neq0$.

To do so, consider $O_1$. By Lemma \ref{lem:NF}, either it is a constant $c$, or a ``fresh'' variable $y_k$. In the first case, build $$t^{(1)}:=\frac{1}{\sqrt{2}^{p}}\sum e^{2i\pi P}\ketbra{O_2,...,O_m}{\vec I},$$ in the second case, build $$t^{(1)}:=\left(\frac{1}{\sqrt{2}^{p}}\sum e^{2i\pi P}\ketbra{O_2,...,O_m}{\vec I}\right)[y_k\leftarrow 0].$$

Notice that:
\begin{itemize}
\item $t^{(1)}$ is irreducible
\item $\Var(P^{(1)})\subseteq \Var(\vec O^{(1)},\vec I^{(1)})$
\item $\interp{t^{(1)}} = (\bra{c}\otimes id_{m-1})\circ\interp{t'}\quad$ if $O_1=c$
\item $\interp{t^{(1)}} = (\bra{0}\otimes id_{m-1})\circ\interp{t'}\quad$ if $O_1=y_k$
\end{itemize}
$c$ (resp.~$0$) will be the first term of $\vec x_1$.

Doing so repeatedly (building $t^{(i+1)}$ from $t^{(i)}$) first for the whole ket, and then for the whole bra, we end up with a term $t^{(n+m)}$ of the form $t^{(n+m)} = \frac{1}{\sqrt{2}^{p}}\sum e^{2i\pi c}$ with $c$ a constant. In the process, we build $\vec x_1$ and $\vec x_2$.

Clearly, $\interp{t^{(n+m)}}\neq 0$, and yet, $\interp{t^{(n+m)}} = \bra{\vec x_1}\interp{t'}\ket{\vec x_2}$. Hence, $\interp{t'}\neq 0$. We end up with a contradiction, so $t$ actually reduces to $\sum\limits_{y_0} e^{2i\pi \frac{y_0}2}\ketbra{0,...,0}{0,...,0}$.
\end{proof}

\begin{proof}[Lemma \ref{lem:polynomial-uniqueness}]
\phantomsection
\label{prf:polynomial-uniqueness}
Let us prove the result by induction on $k$:
\begin{itemize}
\item If $k=0$, the result is obvious
\item Suppose the result is true for $k$. Let $P_i\in\mathbb R[X_1,...,X_{k+1}]/(1,X^2-X)$. Then $P_i(\vec x,x_0) = P_i'(\vec x) + x_0P_i''(\vec x)$ with $P_i',P_i''\in\mathbb R[X_1,...,X_k]/(1,X^2-X)$. By hypothesis, we have $\forall \vec x\in\{0,1\}^k,~P_1(\vec x,0) = P_2(\vec x,0)$, so by induction hypothesis, $P_1'=P_2'$. Similarly, we get $P_1''=P_2''$. Hence, $P_1=P_2$.
\end{itemize}
\end{proof}

\begin{proof}[Theorem \ref{thm:clif-completeness}]
\phantomsection
\label{prf:clif-completeness}
If $\interp{t_1}=0=\interp{t_2}$, by Proposition \ref{prop:zero-morphism}, the two terms reduce to the same normal form.

Suppose now that $\interp{t_i}\neq0$, and that $t_i$ reduce to $t_i'=\frac{1}{\sqrt{2}^{p_i}}\sum e^{2i\pi P_i}\ketbra{\vec O^{(i)}}{\vec I^{(i)}}$, irreducible. By Corollary \ref{cor:non-zero-morphism}, $\Var(P_i)\subseteq\Var(\vec O^{(i)},\vec I^{(i)})$.

We first show that $\vec O^{(1)} = \vec O^{(2)}$ and $\vec I^{(1)} = \vec I^{(2)}$ while at the same time building the $\alpha$-conversion. Consider $O^{(i)}_1$. By Lemma \ref{lem:NF}, either $O^{(i)}_1=c$ constant or $O^{(i)}_1=y_{k_i}$. We can show that $O^{(1)}$ and $O^{(2)}$ are in the form. Indeed, suppose $O^{(1)}_1=c$ and $O^{(2)}_1=y_{k}$. Then, $(\bra{c{\oplus}1}\otimes id)\circ\interp{t_1'}=0$, however $(\bra{c{\oplus}1}\otimes id)\circ\interp{t_2'}=\interp{\left(\frac{1}{\sqrt{2}^{p_i}}\sum e^{2i\pi P_i}\ketbra{O^{(2)}_2,...,O^{(2)}_m}{\vec I^{(i)}}\right)[y_k\leftarrow c{\oplus}1]}\neq0$ by Corollary \ref{cor:non-zero-morphism}, since the last term is irreducible with no internal variable.

Hence, either $O^{(1)}_1 = c = O^{(2)}_1$ or $O^{(1)}_1 = y_{k_1}$ and $O^{(2)}_1 = y_{k_2}$. In the first case, build $$t^{(1)}_i := \frac{1}{\sqrt{2}^{p_i}}\sum e^{2i\pi P_i}\ketbra{O^{(i)}_2,...,O^{(i)}_m}{\vec I^{(i)}}.$$ In the second case, build $$t^{(1)}_i := \left(\frac{1}{\sqrt{2}^{p_i}}\sum e^{2i\pi P_i}\ketbra{O^{(i)}_2,...,O^{(i)}_m}{\vec I^{(i)}}\right)[y_{k_i}\leftarrow 0],$$ and the $\alpha$-conversion $y_{k_1}\leftrightarrow y_{k_2}$.

In parallel, we start building a particular operator that will be of use in the following. In the first case, the operator is built from $\operatorname{op}:=\bra+$, in the second case, from $\operatorname{op}:=id$.

We may notice that:
\begin{itemize}
\item $t^{(1)}_i$ is irreducible
\item $t^{(1)}_i$ has no internal variable
\item by Corollary \ref{cor:non-zero-morphism}, $\interp{t^{(1)}_i}\neq 0$
\item $\interp{t^{(1)}_i} = (\bra c\otimes id)\circ\interp{t_i'}\quad$ if $O^{(i)}_1 = c$
\item $\interp{t^{(1)}_i} = (\bra 0\otimes id)\circ\interp{t_i'}\quad$ if $O^{(i)}_1 = y_{k_i}$
\item $(\operatorname{op}\otimes id)\circ \interp{t_i'} = \interp{\frac{1}{\sqrt{2}^{p_i}}\sum e^{2i\pi P_i}\ketbra{O^{(i)}_2,...,O^{(i)}_m}{\vec I^{(i)}}}\quad$ if $O^{(i)}_1 = c$
\item $(\operatorname{op}\otimes id)\circ \interp{t_i'} = \interp{\frac{1}{\sqrt{2}^{p_i}}\sum e^{2i\pi P_i}\ketbra{y_{k_i},O^{(i)}_2,...,O^{(i)}_m}{\vec I^{(i)}}}\quad$ if $O^{(i)}_1 = y_{k_i}$
\end{itemize}
Doing so inductively first for the whole ket, then for the whole bra, we get:
\begin{itemize}
\item a matching of variables of $t_2'$ with variables of $t_1'$. We may call $\delta$ the bijection that maps a variable of $t_2'$ to a variable of $t_1'$.
\item the equalities $\vec O^{(1)} = \vec O^{(2)}[\vec y_2 \leftarrow \delta(\vec y_2)]$ and $\vec I^{(1)} = \vec I^{(2)}[\vec y_2 \leftarrow \delta(\vec y_2)]$
\item the equality $\frac{1}{\sqrt{2}^{p_1}} e^{2i\pi P_1[\vec y_1\leftarrow \vec 0]} = \interp{t_1^{(n+m)}} = \interp{t_2^{(n+m)}} = \frac{1}{\sqrt{2}^{p_2}} e^{2i\pi P_2[\vec y_2\leftarrow \vec 0]}$ which implies equality for the $p_i$ and the constants in the phase polynomials.
\item two operators $\operatorname{op}_1$ (for the ket) and $\operatorname{op}_2$ (for the bra), such that $$\operatorname{op}_1\circ\interp{t_i'}\circ\operatorname{op}_2 = \interp{\frac{1}{\sqrt{2}^{p_i}}\sum e^{2i\pi P_i}\ketbra{y^{(i)}_1,...}{...,y^{(i)}_k}}$$
\end{itemize}
It remains to show that $P_1 = P_2[\vec y_2\leftarrow \delta(\vec y_2)]$. We have:
\begin{align*}
\sum_{\vec y\in\{0,1\}^k}&e^{2i\pi P_1(\vec y)}\ketbra{y_1,...}{...,y_k}
= \interp{\sum e^{2i\pi P_1}\ketbra{y_1,...}{...,y_k}}\\
&= \operatorname{op}_1\circ\interp{t_1'}\circ\operatorname{op}_2
= \operatorname{op}_1\circ\interp{t_2'[\vec y_2\leftarrow \delta(\vec y_2)]}\circ\operatorname{op}_2\\
&= \interp{\sum e^{2i\pi P_2}\ketbra{y_1,...}{...,y_k}}
= \sum_{\vec y\in\{0,1\}^k}e^{2i\pi P_2(\vec y)}\ketbra{y_1,...}{...,y_k}
\end{align*}
By linear independence of the family $\left(\ketbra{y_1,...}{...,y_k}\right)_{\vec y\in\{0,1\}^k}$, we have:
$$\forall \vec y\in\{0,1\}^k,~ e^{2i\pi P_1(\vec y)}=e^{2i\pi P_2(\vec y)}$$
Since the $P_i$ are considered modulo $1$, we have $\forall \vec y\in\{0,1\}^k,~ P_1(\vec y)= P_2(\vec y)$.
By Lemma \ref{lem:polynomial-uniqueness}, we finally get $P_1=P_2$.
\end{proof}

\begin{proof}[Proposition \ref{prop:clif!-terminates}]
\phantomsection
\label{prf:clif!-terminates}
For a morphism $s\sum\limits_{\vec y} e^{2i\pi P}\ket{\vec O}!\vec D\bra{\vec I}$ of $\cat{SOP}^{\sground}$, consider the tuple:\\
\resizebox{\columnwidth}{!}{$\left(\abs{\vec y}, \abs{\{D_i\in\vec D~|\operatorname{mon}(D_i)\geq2\}}, \sum\limits_i\operatorname{mon}(D_i), \operatorname{mon}(O_1),\ldots, \operatorname{mon}(O_m),\operatorname{mon}(I_1),\ldots,\operatorname{mon}(I_n),\operatorname{mon}(P)\right)$}\\[0.5em]
Again, we define an order on these, as their lexicographic order. We can show that all the rules of $\rewrite{\operatorname{Clif}\raisebox{0.5ex}\sground}$ reduce the tuple:
\begin{itemize}
\item (Elim), (HH), ($\omega$) and (H\sground) all reduce $\abs{\vec y}$
\item (disc) reduces $\abs{\{D_i\in\vec D~|\operatorname{mon}(D_i)\geq2\}}$
\item ($\oplus$\sground) reduces $\sum_i\operatorname{mon}(D_i)$ and sometimes even $\abs{\{D_i\in\vec D~|\operatorname{mon}(D_i)\geq2\}}$
\item (Cst\sground) reduces $\sum_i\operatorname{mon}(D_i)$
\item (ket) reduces $\operatorname{mon}(O_i)$ and none of the previous quantities in the tuple
\item (bra) reduces $\operatorname{mon}(I_i)$ and none of the previous quantities in the tuple
\item (Z\sground) reduces $\operatorname{mon}(P)$ and none of the previous quantities
\item (Z) reduces any tuple with $\abs{\vec y}\geq1$ and $\operatorname{mon}(P)\geq 1$ to $(1,0,\ldots,0,1)$
\end{itemize}
\end{proof}

\begin{proof}[Lemma \ref{lem:clif-disc-form}]
\phantomsection
\label{prf:clif-disc-form}
The first and last conditions are verified just as in the pure case. Then, all the constants in the phase polynomial can be removed using rule (Z\sground).

Then for the form of $\vec D$, let us decompose it as $\vec D = \{y_1,\ldots,y_k\}\cup\{D_{i_1},\ldots,D_{i_s}\}$ where all the polynomials in the right hand side have $\operatorname{mon}(.)\geq 2$ (if $0$ or $1$ appears as a polynomial in $\vec D$, it is removed using (Cst\sground)).
Consider $D_{i_1}$. Either 
\begin{itemize}
\item $D_{i_1}$ contains at least one variable $y_{k+1}\notin\{y_1,\ldots,y_k\}$, in which case (disc) can be used so $\vec D'=\{y_1,\ldots,y_{k+1}\}\cup\{D_{i_2},\ldots,D_{i_s}\}[y_{k+1}\leftarrow D_{i_1}]$
\item or $D_{i_1}$ contains only variables of $\{y_1,\ldots,y_k\}$, in which case, using ($\oplus$\sground) repeatedly, it can be reduced to a constant that can then be removed using (Cst\sground), so $\vec D' = \{y_1,\ldots,y_k\}\cup\{D_{i_2},\ldots,D_{i_s}\}$
\end{itemize}
Hence, in any case, $\vec D$ can be reduced to the form $\vec D = \{y_1,\ldots,y_k\}$.

We then have to show that $P$ can be reduced to the form above. Suppose $y_0$ appears both in $!\{\vec y_d\}$ and in $P^{(1)}$, then (Z\sground) can be used to remove it from $P$. The same goes for monomials of the form $y_0y_0'$ in $P^{(2)}$ when $\{y_0,y_0'\}\subseteq \{\vec y_d\}$.

Finally, if a variable of $\vec y_d$ appears only in $!\{\vec y_d\}$ and in $P$, then the rule (H\sground) can be applied to remove the variable.
\end{proof}

\begin{proof}[Lemma \ref{lem:zero-morphisms-disc}]
\phantomsection
\label{prf:zero-morphisms-disc}
The proof is similar to that of Proposition \ref{prop:zero-morphism}, except now we have a set of discarded variables $\{\vec y_d\}$. However, since $\{\vec y_d\}\subseteq\Var(\vec O,\vec I)$, the set of discarded variables will deplete as the $t^{(i)}$ are built. The conclusion remains unchanged.
\end{proof}

\begin{proof}[Proposition \ref{prop:G-cd}]
\phantomsection
\label{prf:G-cd}
First, let us prove that, if $t\in\overline{\cat{SOP}^{\sground}}$ is terminal, $G(F(t))$ is defined and $G(F(t))=t$. By definition:
$$t = \frac1{\sqrt2^p}\sum_{\vec y,\vec y_d} e^{2i\pi P(\vec y,\vec y_d)}\ket{\vec O(\vec y,\vec y_d)}!\{\vec y_d\}\bra{\vec I(\vec y,\vec y_d)}$$
where $P$ has no constant, all its monomials contain a variable of $\vec y$, $\{\vec y\}\subseteq\Var(\vec O,\vec I)$ and $\{\vec y_d\}\subseteq\Var(\vec O,\vec I,P)$. Hence, again by definition:
$$F(t):= \frac1{\sqrt2^{2p}}\sum_{\vec y,\vec y',\vec y_d} e^{2i\pi\left(P(\vec y,\vec y_d)-P(\vec y',\vec y_d)\right)}\ketbra{\vec O(\vec y,\vec y_d),\vec O(\vec y',\vec y_d)}{\vec I(\vec y,\vec y_d),\vec I(\vec y',\vec y_d)}$$
Notice that:
\begin{itemize}
\item obviously $\abs{\vec O(\vec y,\vec y_d)}=\abs{\vec O(\vec y',\vec y_d)}$ and $\abs{\vec I(\vec y,\vec y_d)}=\abs{\vec I(\vec y',\vec y_d)}$
\item $\{\vec y_d\} = \{\vec y,\vec y',\vec y_d\} \setminus \Var(\vec O(\vec y,\vec y_d)\oplus \vec O(\vec y',\vec y_d), \vec I(\vec y,\vec y_d)\oplus \vec I(\vec y',\vec y_d))$.\\
Indeed, if $O_i(\vec y,\vec y_d) = y_{i_1}\oplus\ldots\oplus y_{i_k}\oplus {y_d}_{j_1}\oplus\ldots\oplus {y_d}_{j_\ell}$, then $O_i(\vec y,\vec y_d)\oplus O_i(\vec y',\vec y_d)= y_{i_1}\oplus\ldots\oplus y_{i_k}\oplus y_{i_1}'\oplus\ldots\oplus y_{i_k}'$, so $\{\vec y_d\}\cap \Var(\vec O(\vec y,\vec y_d)\oplus \vec O(\vec y',\vec y_d),\vec I(\vec y,\vec y_d)\oplus \vec I(\vec y',\vec y_d))=\varnothing$. Moreover, all the variables of $\vec y$ and $\vec y'$ appear somewhere in $\Var(\vec O(\vec y,\vec y_d)\oplus \vec O(\vec y',\vec y_d), \vec I(\vec y,\vec y_d)\oplus \vec I(\vec y',\vec y_d))$, since $\{\vec y\}\subseteq\Var(\vec O(\vec y,\vec y_d),\vec I(\vec y,\vec y_d))$.
\item $\{\vec y\}:= \Var(\vec O(\vec y,\vec y_d),\vec I(\vec y,\vec y_d))\setminus \{\vec y_d\}$ for roughly the same reasons
\item $\{\vec y'\}:=\left(\{\vec y,\vec y',\vec y_d\}\setminus\{\vec y\}\right)\setminus\{\vec y_d\}$
\item by construction of $F$, $\abs{\vec y}=\abs{\vec y'}$
\item whenever $\vec O_i(\vec y,\vec y_d)\oplus \vec O_i(\vec y',\vec y_d) = y_{i_1}\oplus y_{i_2}'$ we define $\delta(y_{i_2}') := y_{i_1}$. We need to show that it completely and uniquely defines $\delta$ as a bijection. Consider the variable $y_i$. Let $K_i$ be the first (from left to right) polynomial of $(\vec O(\vec y,\vec y_d),\vec I(\vec y,\vec y_d))$ where $y_i$ appears. Then $K_i(\vec y,\vec y_d)=y_i$, otherwise, either (ket) or (bra) could be applied on $t$, which means $t$ is not terminal. Hence $K_i(\vec y,\vec y_d)\oplus K_i(\vec y',\vec y_d)=y_i\oplus y_i'$, so $\delta(y_i')=y_i$, and $y_i'$ is the only possible preimage of $y_i$ by $\delta$. Notice that $\delta(\vec y')=\vec y$ with no permutation on the indexes, so we obviously get $(\vec O(\vec y,\vec y_d)\oplus \vec O(\vec y',\vec y_d), \vec I(\vec y,\vec y_d)\oplus \vec I(\vec y',\vec y_d))[\vec y'\leftarrow \delta(\vec y')]=\vec 0$.
\end{itemize}
Hence, $G(F(t))$ is well defined, and:
\begin{align*}
G(F(t)) &= \frac1{\sqrt2^p}\sum_{\vec y,\vec y_d} e^{-2i\pi\left(P(\vec y,\vec y_d)-P(\vec y',\vec y_d)\right)[\vec y\leftarrow \vec 0][\vec y'\leftarrow \vec y]}\left(\ket{\vec O(\vec y',\vec y_d)}!\{\vec y_d\}\bra{I(\vec y',\vec y_d)}\right)[\vec y{\leftarrow} \vec 0][\vec y'{\leftarrow} \vec y]\\
&=\frac1{\sqrt2^p}\sum_{\vec y,\vec y_d} e^{-2i\pi\left(0-P(\vec y,\vec y_d)\right)}\ket{\vec O(\vec y,\vec y_d)}!\{\vec y_d\}\bra{I(\vec y,\vec y_d)}
 = t
\end{align*}

We now need to show that for all the terms $t'$ that are reduced from $F(t)$, $G(t')$ is defined, and $G(t')=G(F(t))=t$. To do so, we show by induction that along any reduction path from $F(t)$, some properties are preserved.

Let $t' = \frac1{\sqrt2^{p'}}\sum\limits_{\vec y^{(t')}} e^{2i\pi P}\ketbra{\vec O_1,\vec O_2}{\vec I_1,\vec I_2}$ such that $F(t)\overset\ast{\rewrite{\operatorname{Clif+}}} t'$. We claim that:
\begin{itemize}
\item $p' = 2p$
\item $\vec y^{(t')} = \vec y,\vec y',\vec y_d$, i.e.~no variable is removed, and the partitioning by $G$ does not change
\item $\Var(\vec O_1,\vec O_2,\vec I_1,\vec I_2) = \{\vec y,\vec y',\vec y_d\}$, i.e.~no variable becomes internal
\item $\forall k,~\abs{\Var(O^{(k)}_i)\cap \{\vec y_d\}}\leq 1~\text{ and }~\begin{cases}~\Var(O^{(k)}_1)\setminus\{\vec y_d\}\subseteq \Var(O^{(1)}_1,...,O^{(k-1)}_1)\\~\Var(O^{(k)}_2)\subseteq \Var(\vec O_1,O^{(1)}_2,...,O^{(k-1)}_2)\\\text{or}\\~O^{(k)}_i = y_{k'}\end{cases}$
\item $\forall k,~\abs{\Var(I^{(k)}_i)\cap \{\vec y_d\}}\leq 1~\text{ and }~\begin{cases}~\Var(I^{(k)}_1)\setminus\{\vec y_d\}\subseteq \Var(\vec O_1,\vec O_2,I^{(1)}_1,...,I^{(k-1)}_1)\\~\Var(I^{(k)}_2)\subseteq \Var(\vec O_1,\vec O_2,\vec I_1,I^{(1)}_2,...,I^{(k-1)}_2)\\\text{or}\\~I^{(k)}_i = y_{k'}\end{cases}$
\item $G(t')$ is well defined and $G(t')=t$
\end{itemize}

It is the case for $F(t)$. Indeed, $t$ is terminal with $\rewrite{\operatorname{Clif}\raisebox{0.5ex}\sground}$, so since the rule (ket) cannot be applied to $t$, it in particular implies the above properties on $F(t)$.

Let us consider one such $t'$. Notice that since there are no internal variables, none of the rules (Elim), (HH), ($\omega$) or (Z) can be applied. Suppose $t'\rewrite{\operatorname{Clif+}}t''$ in one step. Only (ket) or (bra) can be applied from $t'$ to $t''$, and only on either $\vec O_1$ or $\vec I_1$. Without loss of generality, suppose (ket) is applied on, in the polynomial $O_1^{(k)}$ of $t'$. Notice that $O_1^{(k)}$ is necessarily of the form $O_1^{(k)} = y_{d_k}\oplus O'$ where $y_{d_k}\in\{\vec y_d\}\setminus\Var(O^{(1)}_1,...,O^{(k-1)}_1)$ and $\Var(O')\subseteq\Var(O^{(1)}_1,...,O^{(k-1)}_1)$, otherwise the rule cannot be applied.

By application of the rule, $t'' = t'[y_{d_k}{\leftarrow}y_{d_k}{\oplus}O']$. Since $y_{d_k}\notin\Var(O^{(1)}_1,...,O^{(k-1)}_1)$, the first $k-1$ polynomials in the ket are left unchanged, so the variables $\Var(O^{(1)}_1,...,O^{(k-1)}_1)$ are still in the ket. In particular, the variables $\Var(O')\subseteq\Var(O^{(1)}_1,...,O^{(k-1)}_1)$ are also still present in the ket. The substitution cannot remove other variables from the ket, so $\Var\left(\vec O_i[y_{d_k}{\leftarrow}y_{d_k}{\oplus}O'],\vec I_i[y_{d_k}{\leftarrow}y_{d_k}{\oplus}O']\right) = \{\vec y,\vec y',\vec y_d\}$. The overall scalar is obviously unchanged. It is fairly easy to check the property of the ket for $t''$. Finally, we can show that the partitioning of variables by $G$ is unchanged.

We assume that $G(t')$ is well defined. $\vec O_1$ and $\vec O_2$ are hence of the same size. $\{\vec y_d\}$ is defined for $t'$ as $\{\vec y_d\} = \{\vec y^{(t')}\}\setminus\Var(\vec O_1\oplus\vec O_2)$. Notice that $y_{d_k} \in\{\vec y_d\}$ so $y_{d_k}\notin \Var(\vec O_1\oplus\vec O_2)$. Hence, if it appears somewhere in $\vec O_1$, say in $O_1^{(j)}$, it also appears in $O_2^{(j)}$ so that $y_{d_k}\notin \Var(O_1^{(j)}\oplus O_2^{(j)})$. So, the substitution will not change:
$$\{\vec y_d\} = \{\vec y^{(t')}\}\setminus\Var(\vec O_1\oplus\vec O_2) = \{\vec y^{(t')}\}\setminus\Var(\vec O_1[y_{d_k}{\leftarrow}y_{d_k}{\oplus}O']\oplus\vec O_2[y_{d_k}{\leftarrow}y_{d_k}{\oplus}O'])$$
Similarly, $\{\vec y_1\}$ and $\{\vec y_2\}$ are left unchanged, as well as the bijection $\delta$.\\
Since $\Var(O')\subseteq\Var(O^{(1)}_1,...,O^{(k-1)}_1)\subseteq\{\vec y_1\}$, the substitution $[y_{d_k}{\leftarrow}y_{d_k}{\oplus}O'][\vec y_1\leftarrow \vec 0][\vec y_2\leftarrow\delta(\vec y_2)]$ is the same as the substitution $[\vec y_1\leftarrow \vec 0][\vec y_2\leftarrow\delta(\vec y_2)]$. Hence, $G(t'')=G(t')=t$.

The whole reasoning is similar when the rule (bra) is applied instead of (ket). This concludes the proof.
\end{proof}

\end{document}

%% file: figures/H-spider-decomp-proof/H-spider-decomp-proof_00.tikz
\begin{tikzpicture}
	\begin{pgfonlayer}{nodelayer}
		\node [style=box] (0)  at (0.25, 0.5) {$r$};
		\node [style=box] (1)  at (0.25, 0.0) {};
		\node [style=box] (2)  at (-0.25, 0.0) {$\frac12$};
		\node [style=none] (3)  at (0.25, -0.5) {};
	\end{pgfonlayer}
	\begin{pgfonlayer}{edgelayer}
		\draw (0) to (3.center);
	\end{pgfonlayer}
\end{tikzpicture}

%% file: figures/H-spider-decomp-proof/H-spider-decomp-proof_01.tikz
\begin{tikzpicture}
	\begin{pgfonlayer}{nodelayer}
		\node [style=box] (5)  at (0.375, 0.375) {$\rho e^{i\theta}$};
		\node [style=box] (7)  at (-0.375, 0.125) {$\frac{1+r}2$};
		\node [style=none] (8)  at (0.375, -0.375) {};
	\end{pgfonlayer}
	\begin{pgfonlayer}{edgelayer}
		\draw (5) to (8.center);
	\end{pgfonlayer}
\end{tikzpicture}

%% file: figures/H-spider-decomp-proof/H-spider-decomp-proof_02.tikz
\begin{tikzpicture}
	\begin{pgfonlayer}{nodelayer}
		\node [style=box] (10)  at (0.0, 0.625) {$e^{i\alpha}$};
		\node [style=box] (11)  at (-0.5, -0.125) {$s$};
		\node [style=none] (12)  at (0.0, -0.625) {};
		\node [style=box] (13)  at (0.0, 0.25) {};
		\node [style=white dot] (14)  at (0.0, -0.125) {};
		\node [style=box] (15)  at (0.5, 0.125) {$e^{i\beta}$};
	\end{pgfonlayer}
	\begin{pgfonlayer}{edgelayer}
		\draw (10) to (12.center);
		\draw (14) to (15);
	\end{pgfonlayer}
\end{tikzpicture}